\documentclass[12pt]{article}

\usepackage[T1]{fontenc}
\usepackage{amsmath,amsfonts,amsthm,amssymb,mathtools}
\usepackage{tikz}
\usetikzlibrary{plotmarks}
\usepackage{pgfplots,pgfplotstable}
\pgfplotsset{compat=1.7}
\usetikzlibrary{matrix}
\usepackage{ytableau}
\usepackage{hyperref}

\usepackage{authblk}
\usepackage{graphicx}
\usepackage{natbib}
\usepackage{array}
\usepackage[normalem]{ulem}
\usepackage{stmaryrd}
\usepackage{tikz}
\usetikzlibrary{plotmarks}
\usepackage{pgfplots,pgfplotstable}
\usepackage{ytableau}
\usepackage{xcolor}
\usepackage{xspace}

\newtheorem{theorem}{Theorem}
\newtheorem{proposition}[theorem]{Proposition}
\newtheorem{lemma}[theorem]{Lemma}
\newtheorem{corollary}[theorem]{Corollary}
\newtheorem{definition}{Definition}
\newtheorem{remark}{Remark}

\newtheorem{modeling}{Modeling}
\newtheorem{example}{Example}

\DeclareMathOperator{\SM}{{SM}}
\DeclareMathOperator{\HS}{{HS}}
\DeclareMathOperator{\rk}{{rk}}
\DeclareMathOperator{\stab}{{stab}}
\newcommand{\plucker}{\mathcal P}
\newcommand{\idealSM}{\mathcal{I}} 
\newcommand{\idseq}{\mathcal I_{eq}}

\newcommand{\Seq}{{\mathcal{S}_{eq}}}

\newcommand{\piC}{\Pi_{C}}
\newcommand{\piU}{\Pi_{U}}
\newcommand{\piS}{\Pi_{S}}

\newcommand{\minor}[2]{\left\vert #1 \right\vert_{#2}}
\newcommand{\twistedbinom}[2]{
\left[\substack{#1\\#2}\right]
}

\newcommand{\eqdef}{\stackrel{\text{def}}{=}}
\newcommand{\any}{*}
\renewcommand{\vec}[1]{\ensuremath{\mathbf{#1}}}
\newcommand{\zerom}{\vec{0}}
\newcommand{\mat}[1]{\boldsymbol{#1}}

\newcommand{\Cm}{{\mat C}}

\newcommand{\Fm}{{\mat F}}
\newcommand{\Um}{{\mat U}}
\newcommand{\Ym}{{\mat Y}}
\newcommand{\xv}{{\mat X}}
\newcommand{\ev}{{\mat e}}
\newcommand{\fq}{\mathbb K}

\begin{document}
\title{Computation of the Hilbert Series for the  Support-Minors Modeling of the MinRank Problem}
             
\author[1]{Magali Bardet\thanks{magali.bardet@univ-rouen.fr}}
\author[1]{Alban Gilard\thanks{alban.gilard@univ-rouen.fr}}
\affil[1]{Université de Rouen Normandie, LITIS UR 4108}

\maketitle
     
\begin{abstract}
  The MinRank problem is a simple linear algebra problem: given
  matrices with coefficients in a field, find a non trivial linear
  combination of the matrices that has a small rank.

  There are several algebraic modelings of the problem. The main ones
  are: the Kipnis-Shamir modeling, the Minors modeling and the
  Support-Minors modeling. The Minors modeling has been studied by
  Faugère et al. (2010), 
  where the authors provide an analysis of the complexity of computing
  a Gröbner basis of the modeling, through the computation of the
  exact Hilbert Series for a generic instance. For the Support-Minors
  modeling, the first terms of the Hilbert Series are given by Bardet
  et al. (2020) 
  based on heuristic and experimental work.

  In this work, we provide a formula and a proof for the complete
  Hilbert Series of the Support Minors modeling for generic
  instances. This is done by adapting well known results on
  determinantal ideals to an ideal generated by a particular subset of
  the set of all minors of a matrix of variables. We show that this
  ideal is generated by standard monomials having a particular
  shape, and derive the Hilbert Series by counting the number of such
  standard monomials.

  This work allows to make a precise comparison between the Minors and
  Support Minors modeling, and a precise estimate of the complexity of
  solving MinRank instances for the parameters of the Mirath signature
  scheme that is currently at the second round of the NIST
  standardization process for Additional Digital Signature Schemes.
\end{abstract}

\paragraph{keyword}{MinRank, Support Minors modeling, Determinantal ideals, Standard monomials, Hilbert series, Gröbner bases, Multivariate cryptography}
\section{Introduction}
\label{sec:intro}
The MinRank problem is a very simple and classical linear algebra
problem: find a non-trivial linear combination of given matrices that
has a small rank. This problem has been studied for years: its
NP-hardness has been proven in \citet{BFS99}.
It has many applications in various fields
(e.g. robotics, real geometry) and plays a central role in public key
cryptography, especially since the beginning of the NIST Post-Quantum
Standardization
Process\footnote{\url{https://csrc.nist.gov/pqc-standardization}}. It
was for instance used by \citet{DCPSYKP20} to attack Rainbow, a
signature scheme that was a finalist at the third round of the NIST
call. It is exactly the decoding problem for matrix codes in
rank-metric code-based cryptography. The security of the MIRA
\cite{mira} and MiRitH \cite{mirith} signature schemes, that have
merged to Mirath \cite{mirath} for the second round of the additional call for
Digital Signature
schemes\footnote{\url{https://csrc.nist.gov/Projects/pqc-dig-sig/}},
is based on the hardness of solving uniformly random instances of
MinRank.  Hence, analyzing the complexity of solving the problem is of
greatest importance, in particular for generic instances.

We focus in this paper on the algebraic modelings of the problem. The MinRank problem can be rephrased as the problem of finding the set of points at which a matrix, whose entries are linear forms, has a small rank. 
We analyze directly the {\it Generalized MinRank Problem} (GMR), where the entries of the matrix are homogeneous polynomials of some degree $D$.

\begin{definition}[Homogeneous Generalized MinRank Problem]
  Let $\fq$ be a field, $r$ and $D$ two integers, and $\Fm$ a $m \times n$ matrix
  \begin{align*}
    \Fm = \begin{pmatrix}
      f_{1, 1} & \cdots & f_{1, n} \\
      \vdots & & \vdots \\
      f_{m, 1} & \cdots & f_{m, n}
    \end{pmatrix}.
  \end{align*}
  where $f_{i, j}$ is an homogeneous polynomial in
  $\fq[x_1,\dots, x_K]$ of degree $D$. We want to compute the set of
  points in the algebraic closure $\overline{\fq}$ of $\fq$ at which
  the evaluation of $\Fm$ has rank at most $r$.
\end{definition}

This set of points can be described by a Gröbner basis of the ideal generated by any algebraic modeling of the problem. It is well known since the work from~\citet{L83} that a Gröbner basis can be computed by linear algebra on a finite number of Macaulay matrices\footnote{A Macaulay matrix in degree $d_x$ represents a  basis of the augmented $\fq$-vector space in degree $d_x$ obtained by multiplying all the equations from the algebraic system by all the possible monomials in $x$ to get polynomials of degree $d_x$.}. The Hilbert series of an ideal provides the rank of all these Macaulay matrices, and knowing these values allows to predict an upper bound on the complexity of the computation.
Note that for cryptographic applications, this set of points reduces most of the time to only one (projective) point, whose coordinates are in $\fq$.

\paragraph{Previous work}
The three main algebraic modelings for the GMR Problem are the
Kipnis-Shamir (KS) modeling~\cite{KS99}, the Minors modeling~\cite{FSS10}
and the Support Minors (SM) modeling~\cite{BBCGPSTV20}.

The Kipnis-Shamir modeling is constructed from the fact that
$\rk(\Fm)\le r$ if and only if its kernel contains at least $n-r$
linearly independent vectors. This modeling is intrinsically affine,
and the complexity of solving the (KS) algebraic system is not well
understood and seems hard to predict. Indeed, it has been shown independently by
\citet{BB22} and \citet{GD22} that the
ideal generated by the (KS) system is equal to the ideal generated by
the affine version of the (SM) system, and that both the Minors and
(SM) equations are produced after a degree fall in degree $r+2$ during a computation of
the Gröbner basis on the (KS) system for a graded monomial ordering (see~\cite[Remark 1]{BB22}).

The Minors modeling is obtained by considering the system of all the
minors of $\Fm$ of size $r+1$, whose associated variety is the set of
solutions of the MinRank problem.  The Minors system of equations has
been thoroughly analyzed by \citet{FSS10,FSS13},
where the authors provide an analysis of the complexity of computing a
Gröbner basis of the modeling, through the computation of the exact
Hilbert series for a generic homogeneous instance. From the Hilbert
series, assuming a variant of Fröberg's conjecture for the
overdetermined case, it is possible to derive for instance the exact
degree of regularity of the Minors system for a generic overdetermined
MinRank problem (that is the use-case in cryptography), and to give a
complexity estimate for the cost of computing the Gröbner basis and
the solutions.

The origin of the Support Minors
modeling comes from the fact that if the matrix of unknowns $\Cm$ of
size $r\times n$ represents a basis of the row space of $\Fm$, then any row
of $\Fm$ is linearly dependent from the rows of $\Cm$, i.e. the
matrices $\Cm_\ell \eqdef
\begin{pmatrix}
\Cm\\f_{\ell,1} \;\dots\;f_{\ell,n}
\end{pmatrix}$ are of rank at most $r$ for all $\ell\in\{1,\dots,m\}$. This is equivalent to the fact that all maximal minors of those matrices are zero.
The main idea behind the Support Minors modeling consists in applying
Laplace expansion along the last rows, and making a change of variables for
the Plücker coordinates $c_I=\minor{\Cm}{\any,I}$ (where $\minor{\Cm}{\any,I}$ is the maximal minor of $\Cm$ over the columns in a set $I$). The interest of
such a change of variable is to reduce the complexity of computing a
Gröbner basis of the system by a factor $r!$, as we replace minors of
$\Cm$ with $r!$ coefficients by a new variable. This leads to the system of polynomials \eqref{eq:SM} described below.
\begin{modeling}[Support Minors Modeling (SM)~\cite{BBCGPSTV20}]
  Let $\Fm\in \fq[\xv]^{m\times n}$ be a Generalized MinRank
  instance with degree $D$ and target rank $r$. The GMR problem can be solved by
  finding $x_1,\dots,x_K\in \fq^K$, and $(c_{I})_{I\subset\lbrace 1,\dots,n\rbrace, \#I=r}\subset \fq^{\binom{n}{r}}$
  such that
  \begin{align}\label{eq:SM}
    \left\lbrace
    \sum_{i=1}^{r+1} (-1)^i f_{\ell,j_i} c_{J\setminus\lbrace j_i\rbrace}=\zerom, \; \substack{\forall J=\{j_1,\dots,j_{r+1}\}\subset\{1,\dots,n\}, 1\le j_1 < \dots < j_{r+1}\le n, \forall \ell\in\lbrace 1,\dots,m\rbrace } \right\rbrace.
  \end{align}
  The $m\binom{n}{r+1}$ equations are bi-homogeneous of bi-degree $(D,1)$ in the $K$ \emph{linear
    variables} $\xv = (x_1,\dots,x_K)$ and the $\binom{n}{r}$
  \emph{minor variables} $c_{I}$, 
  for all $I\subset\lbrace 1,\dots,n\rbrace, \#I=r$.  
\end{modeling}
\citet{BBCGPSTV20} provide a specific Gröbner basis
  algorithm for the SM modeling, that in the overdetermined case
  computes the unique solution of the problem by linear algebra on the
  Macaulay matrices in bi-degree $(d_x,1)$ in the $\xv$ and $c_I$ variables. They give the first terms of the Hilbert series (hence the rank of the Macaulay matrices) up to degree $r+1$ in $\xv$. The result is based on a heuristic analysis, and strengthened by experimental validation.

All the polynomials we consider here are homogeneous, and we are interested in generic systems.
\begin{definition}[Genericity]
  A property is said to be generic if it is true over a non-empty Zariski open subset.
For infinite fields $\fq$, non-empty Zariski open subsets are dense for the Zariski topology.
\end{definition}
In our context, a property on $\Fm$ is generic if  there exists a non-zero multivariate polynomial $h$ such that the property is true for any instance $\Fm$ such that $h$ does  not vanish on the coefficients of the polynomials $f_{i,j}$ in $\Fm$.
\paragraph{Main results}
In this paper, we analyze the ideal $\idealSM$ generated by the Support Minors system in the subalgebra $\fq[\xv][\Cm_I]$ of $\fq[\xv,\Cm]$ generated by all the maximal minors $(\Cm_I)_{I\subset\{1,\dots,n\}, \#I=r}$ of the matrix $\Cm$.
For each integer $d_c\ge 1$, we consider the $\fq[\xv]$-module $\fq[\xv]_{d_c}$ generated by the set of polynomials in $\fq[\xv,\Cm_I]$ of degree exactly $d_c$ in the $\Cm_I$'s, and $\idealSM_{d_c}=\idealSM\cap\fq[\xv]_{d_c}$ the submodule of $\idealSM$ generated by the Support Minors equations in degree $d_c$. Our main result is  the computation of the Hilbert series of $\idealSM_{d_c}$ for the GMR Support Minors system when the matrix $\Fm$ is generic:
\begin{align}
  \HS_{\fq[\xv]_{d_c} / \idealSM_{d_c}}(t) =\left[ \frac{\det(A_{d_c}(t^D))(1-t^D)^{(m-r)(n-r)}}{t^{D\binom{r}{2}} (1-t)^{K}}\right]_+\label{eq:HSAfinal}
\end{align}  
where
$A_{d_c}(t) = \left( \sum\limits_{\ell \geq 0}
  \binom{n+d_c-i}{\ell+d_c} \binom{m-d_c-j}{\ell} t^\ell \right)_{1
  \leq i, j \leq r}$ and the notation $[S(t)]_+$ stands for the power
series obtained by truncating a power series $S(t)\in\mathbb Z[[t]]$
at its first non-positive coefficient.  We prove that this result is
generic for all $K\ge m(n-r)$.  For smaller values of $K$, the
genericity of the result depends on a variant of Fröberg conjecture.
Following our experiments on small values of the parameters for $D=1$, we
  believe this result is still true for any $K$, provided that
  $d_c\le m-r$ and that the series in~\eqref{eq:HSAfinal} is
  truncated. It seems experimentally to be also true for $d_c>m-r$, except when
  $K\in\{(m-r)(n-r),\dots,m(n-r)-1\}$, but we do not have an explanation
  for this at the moment.

For cryptographic application, the instances are taken with coefficients in a finite field. In this case, our experiments indicates that for a large enough finite field, we can expect that the probability that the Hilbert series of a system is~\eqref{eq:HSAfinal} is large, and tends to 1 when the size of the field tends to infinity. We give examples in Section~\ref{sec:experiments}.

\begin{remark}
  Note that our formula~\eqref{eq:HSAfinal} extends for any $d_c\ge 1$ the formula for the Hilbert series of the
  Minors system from \cite{FSS13} ($d_c=0$). For $r=0$ we get the Hilbert
  series of a generic system of $mn$ equations of degree $D$ in $K$
  variables. 
\end{remark}
To obtain this result, we adapt the work from~\citet{FSS13} to
our context. The ideal $\idealSM$ in $\fq[\xv,\Cm_I]$ generated by the maximal
minors of all the matrices $\Cm_\ell$ for $\ell\in\{1,\dots,m\}$ is exactly
the ideal generated by the minors of size $r+1$ of the matrix
$\binom{\Cm}{\Fm}$ that contain the first $r$  rows. This leads us to study the properties 
of particular determinantal ideals generated by minors of a matrix $\binom{\Cm}{\Um}$
that contain the first $r$ rows, where the entries of $\Um$ are variables.
As far as we know, properties of such ideals have not been studied up to now, and we call them \emph{determinantal Support Minors ideals}.
Let  $\mathcal S_{d_c}$ be the $\fq[\Um]$-module of the polynomials in this ideal that have degree $d_c$ in the $\Cm_I$'s. 
We show that the Hilbert series of $\mathcal S_{d_c}$ is given by
\begin{align}
  \HS_{\fq[\Um]_{d_c} / \mathcal{S}_{d_c}}(t) = \frac{\det(A_{d_c}(t))}{t^{\binom{r}{2}} (1-t)^{(m+n-r)r}}\label{eq:HSdcU}
\end{align}  
with the same matrix $A$ as above. The result evaluates to the well known Hilbert series for determinantal rings (see~\cite{BCRV22} for instance) for $d_c=0$.

\paragraph{Organization of the paper}
Section~\ref{sec:SMHSdet} is devoted to the computation of the Hilbert
series~\eqref{eq:HSdcU} for determinantal Support Minors ideals (i.e. for the matrix $\binom{\Cm}{\Um}$). This is done by showing that the module $\mathcal S_{d_c}$ is generated by standard bitableaux with a specific shape. All useful definitions and results on standard bitableaux are recalled in Section~\ref{sec:preliminaries}. 

In Section~\ref{sec:SMHSgeneral}, we show how the properties of determinantal Support Minors ideals with variables can be transferred to the ideal $\idealSM$ for the system (SM) over a non-empty  Zariski open set, as long as $K\ge m(n-r)$. The main contribution of this section is the proof that $\fq[\Um, \Cm_I] / \mathcal S$ is Cohen-Macaulay, where
$\mathcal S=\oplus_{d_c\ge 0} \mathcal S_{d_c}$ and $\mathcal S_0=\mathcal D$
is the determinantal Minors ideal.

Section~\ref{sec:complexity} contains a complexity analysis of the Support Minors system for cryptographic applications using the results from this paper. In particular,  the Mirath parameters over $\mathbb F_{16}$ where chosen according to the best currently known attacks. However, for the Support Minors modeling, only the cases $d_c=1$ and $d_x\le r+1$ where taken into account, as no results where known in the other cases. We are now able to compute as well an upper bound for any $(d_x,d_c)$ and the resulting complexities are almost the same, even if we obtain slightly better results with our new complexity estimates for $d_x\ge r+2$. This reinforces the security analysis and confirms the choice of the Mirath parameters.
We also show that the chosen value of $r$ is quite optimal.

Finally section~\ref{sec:experiments} discusses experimental results to comfort our Fröberg-like genericity conjectures.
\section{Preliminaries}
\label{sec:preliminaries}
We provide in table~\ref{table:notations} a glossary of all structures, sets and ideals used in this paper.
In all the paper we consider homogeneous polynomials. Indeed, theoretical results are much easier to obtain in the homogeneous case. Moreover, as done in~\cite{BBCGPSTV20}, the analysis in the affine case can be reduced to this case by homogenizing the entries of $\Fm$.
\subsection{General Notation}
We denote by $\binom{n}{m} = \frac{n!}{m!(n-m)!}$ the classical
binomial coefficient, and by $\twistedbinom{n}{m}=\binom{n+m}{m}$ the
twisted binomial coefficient, for any integers $m,n\in\mathbb N$. Both are defined to be zero for $m<0$. We also use  the extended
definition of the binomial coefficient $\binom{a}{k}=\frac{a(a-1)\dots(a-k+1)}{k!}$ when $k\in\mathbb N$ and $\alpha$ is any complex number.

Let $\fq$ be a field.
By a slight abuse of notation, we will denote by $\Ym$ both a matrix of
unknowns $\Ym=(y_{i,j})$ and the set of unknowns $\{y_{i,j}\}$. Then, for any matrix of unknowns $\Ym$ of size $m\times n$,  $\fq[\Ym]$ is the polynomial algebra in the $mn$ variables $(y_{i,j})$.
\subsection{Standard monomials}
We recall in
this section the definitions and properties of standard monomials, that will be used in the next section to compute Hilbert series. We refer to~\cite[Chapter 3]{BCRV22} for more details.
Let us consider the set of variables $\{u_{i, j}\}_{i\in\{1,\dots, m\}, j\in\{1,\dots, n\}}$ and the matrix  of unknowns
\begin{align*}
  \Um = (u_{i,j}) = \begin{pmatrix}
    u_{1, 1} & \cdots & u_{1, n} \\
    \vdots & & \vdots \\
    u_{m, 1} & \cdots & u_{m, n}
  \end{pmatrix}.
\end{align*}
The  minors of $\Um$ can be represented as   bivectors
\begin{align*}
  (a | b) = (a_p, \dots , a_1 | b_1, \dots , b_p),
\end{align*}
where $1 \leq a_1 < \dots  < a_p \leq m$ and $1 \leq b_1 < \dots  < b_p \leq n$ represent respectively the rows and columns indexes of $\Um$ which define a minor of size $p$. We call $p$ the \emph{length} of $(a|b)$. Note that for all $i\in\{1,\dots,p\}$, we have $a_i\ge i$ and $b_i\ge i$.
	
We can define a partial order on the set of bivectors (and so on the set of minors of $\Um$) by saying that $(a_p, \dots , a_1 | b_1, \dots , b_p) \leq (\alpha_s, \dots , \alpha_1 | \beta_1, \dots , \beta_s)$ if and only if:
\begin{enumerate}
\item $p \geq s$, and
\item  $a_i \leq \alpha_i$ and $b_i \leq \beta_i$ for all $1 \leq i \leq s$.
\end{enumerate}
We give an example of ordering in figure~\ref{fig:hasse} in the appendix.	
\begin{definition}[Standard monomial, standard bitableau]
  The product $Y = \gamma_1 \dots  \gamma_t$ of $t$ minors of $\Um$ such that $\gamma_i \leq \gamma_{i+1}$ for all $i\in\{1,\dots, t-1\}$ is called a \emph{standard monomial} of degree $d = p_1 + \dots  + p_t$, where $p_i$ is the length of $\gamma_i$. We can identify these standard monomials with a more visual representation in the form of a bitableau by writing vertically each bivector $(a_{i,p_i},\dots,a_{i,1} |b_{i,1},\dots,b_{i,p_i})$ associated to each $\gamma_i$:
  \setlength{\arrayrulewidth}{0.04em}
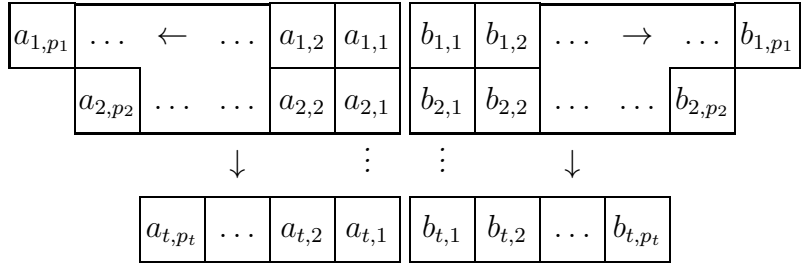
  \begin{align*}
    \ytableausetup{mathmode, boxframe=normal, boxsize=2em}
    \begin{ytableau}
      \cline{2-4}
      a_{1, p_1} & \none[\dots] & \none[\leftarrow] & \none[\dots] & a_{1, 2} & a_{1, 1}\\
      \none & a_{2, p_2} & \none[\dots] & \none[\dots] & a_{2, 2} & a_{2, 1}\\
      \cline{2-6}
      \none &  \none & \none & \none[\downarrow] &\none & \none[\vdots]\\
      \none & \none & a_{t, p_t} & \dots & a_{t, 2} & a_{t, 1}\\
    \end{ytableau} \;
    \ytableausetup{mathmode, boxframe=normal, boxsize=2em}
    \begin{ytableau}
      \cline{3-5}
      b_{1, 1} & b_{1, 2} & \none[\dots] & \none[\rightarrow] & \none[\dots] & b_{1, p_1}\\
      b_{2, 1} & b_{2, 2} & \none[\dots] & \none[\dots] & b_{2, p_2} \\
      \cline{1-5}
      \none[\vdots] & \none & \none[\downarrow] &  \none\\
      b_{t, 1} & b_{t, 2} & \dots & b_{t, p_t}\\
    \end{ytableau}	
  \end{align*}
  The arrows denote the direction of increase of the coefficients.
  Such a bitableau, filled to represent a standard monomial, is called a standard bitableau.
  We define the shape\footnote{We take the definition of shape in~\cite{G94} rather than the one in~\cite{BCRV22} to get the formula in Proposition~\ref{prop:stab}, but it is equivalent.} of $Y$
  as the vector $v = (v(1), \dots , v(p_1))$ such that
  $v(i) = \#\{j : p_j \geq i\}$, and its length as $p_1$. The integer
  $d = v(1) + v(2) + \dots  + v(p_1)$ is the degree of $Y$. We denote
  by $v \overset{p_1}{\rightsquigarrow} d$ the set of all standard
  bitableaux of length $p_1$ and degree $d$, i.e. the tuples
  $v=(v(1),\dots,v(p_1))$ such that $d\ge v(1)\ge \dots\ge v(p_1)$ and $\sum_{i=1}^{p_1}v(i)=d$.
\end{definition}
	
\begin{example}
  For $m = 5$, $n = 4$, the bitableau
  \begin{align*}
    \ytableausetup{mathmode, boxframe=normal, boxsize=2em}
    \begin{ytableau}
      5 & 3 & 2 & 1 \\
      \none & 4 & 3 & 1 \\
      \none & \none & \none & 5 \\
    \end{ytableau} \;
    \ytableausetup{mathmode, boxframe=normal, boxsize=2em}
    \begin{ytableau}
      1 & 2 & 3 & 4\\
      1 & 2 & 3 \\
      2 \\
    \end{ytableau}	
  \end{align*}
  is standard of shape $(3, 2, 2, 1)$, length $4$ and degree $8$. 
\end{example}
	
In the next section, we will construct a basis of the determinantal $\SM$
module using these standard bitableaux, which will be possible mainly
thanks to the following theorem.
\begin{theorem}[Straightening Law { \cite[e.g.][p.72]{BCRV22}}]
  \label{thm:straightening}
  We have the following statements:
  \begin{enumerate}
  \item The standard bitableaux form a basis of $\fq[\Um]$ as a
    $\fq$-vector space.
  \item \label{item:2} If $\gamma$ and $\delta$ are two minors of $\Um$ such that
    $\gamma \delta$ is not standard, then we can write
    $\gamma \delta = \sum_{i} z_i \epsilon_i \eta_i$, with for all
    $i$, $z_i \in \fq$, $\epsilon_i < \gamma$, $\eta_i > \delta$ and
    $\epsilon_i \eta_i$ is standard ($\eta_i$ may be $(|)=1$).
  \item \label{item:3} Let $Y = \delta_1 \dots \delta_t$ be a non-standard
    bitableau. Then we can recover the expression of $Y$ in the basis
    of the standard bitableaux by applying successively the
    straightening relations in~\eqref{item:2}.
  \item \label{item:4} Let $Y$ be a bitableau, and $\gamma_1\dots\gamma_t$ a standard bitableau appearing in the standard representation of $Y$. Then $\gamma_1\le \delta$ for all factors $\delta$ of $Y$.
  \end{enumerate}
\end{theorem}
\begin{remark}\label{remark:li}
  The sum in \eqref{item:2} is finite, as the number of
  bivectors is finite. Note that the theorem remains true if we replace $\fq$ by an arbitrary commutative ring: 
  as explained in~\cite[Remark 3.2.9 p. 78]{BCRV22}, the standard bitableaux indeed generate $\mathbb Z[\Um]$ as a $\mathbb Z$-module, and remain linearly independent over any ring $R[\Um]=R\otimes_{\mathbb Z}\mathbb Z[\Um]$.
  We will use it later over a polynomial ring $\fq[\Ym]$ or $\fq[\Ym, \Cm_I]$, for some unknowns $\Ym$.
\end{remark}
The enumeration of standard tableaux is highly studied in combinatorics and, in particular, we have an explicit formula for the number of standard tableaux with a given shape.
\begin{proposition}[{\cite[e.g.][§14.3]{G94}}]
  \label{prop:stab}
  The number of standard tableaux of shape $v = (v(1), \dots , v(p))$, with $0 \leq v(p) \leq \cdots \leq v(1)$, whose coefficients are bounded by $m$ is given by the formula :
  \begin{align}
  		\stab(m, v) = \det \left(
  		\begin{bmatrix}
  			m-j \\
  			v(i)+j-i
  		\end{bmatrix} \right)_{1 \leq i, j \leq p}\label{eq:stab}
  \end{align}
  where the entries of this $p \times p$ matrix are twisted binomial coefficients.
\end{proposition}
Then for any matrix $\Um$ of size $m\times n$, the number of standard
bitableaux of shape $v$ on the left part and $w$ on the right part
will be the product $\stab(m,v)\stab(n,w)$.

We will use algebras different from $\fq[\Um]$ but with very similar properties, in which the straightening law is always true.
\begin{definition}[Algebra with straightening law, {\cite[e.g.][§4A p. 38]{BV06}}]
Let $B$ be a ring,  $A$ a graded $B$-algebra, and $\Pi \subset A$ a finite subset with partial order $\leq$.
	$A$ is a \emph{graded algebra with straightening law} (ASL) on $\Pi$ over $B$ if the following conditions hold :
	\begin{enumerate}
		\item $A = \bigoplus_{i \geq 0} A_i$ is a graded $B$-algebra such that $A_0 = B$, $\Pi$ consists of homogeneous elements of positive degree and generates $A$ as a $B$-algebra.
		\item The products $\gamma_1 \cdots \gamma_t$, $t \in \mathbb N$, $\gamma_i \in \Pi$, such that $\gamma_1 \leq \cdots \leq \gamma_t$, are linearly independent. They are called standard monomials.
		\item (Straightening law) For all incomparable $\alpha, \beta \in \Pi$ the product $\alpha \beta$ has a representation \[ \alpha \beta = \sum a_{\mu} \mu, \hspace{0.5cm} a_{\mu} \in B, \; a_{\mu} \neq 0, \; \mu \; \text{standard monomial} \] satisfying the following condition: every $\mu$ contains a factor $\gamma \in \Pi$ such that $\gamma < \alpha, \beta$.
	\end{enumerate}	 
\end{definition}

\subsection{Plücker algebra}
\label{sec:plucker}
For a matrix $\Cm\in\fq^{r \times n}$ of unknowns with $n\ge r$, for
any subset $I\subset\{1,\dots,n\}$ of size $r$, we denote by $\Cm_I$ the
maximal minor of $\Cm$ with columns in $I$, and by $c_I$ a variable
representing this polynomial. The \emph{Plücker algebra} is the
subalgebra of $\fq[\Cm]$ given by
$\fq[(\Cm_I)_{I\subset\{1,\dots,n\},\#I=r}]=\fq[\Cm_I]$ for short. It is the homogeneous coordinate ring of the Grassmannian variety parametrizing $r$-dimensional vector subspaces of $\fq^n$. The
Plücker algebra can also be viewed as the quotient
\begin{align}
  \label{eq:plucker}
\fq[c_I]/\plucker \simeq \fq[\Cm_I]
\end{align}
where $\plucker$ is the Plücker ideal, which
is the kernel of the map $\fq[c_I]\to\fq[\Cm_I] : c_I \mapsto \Cm_I$
and is generated by the so called Plücker relations (see for
instance~\cite[Corollary 3.2.7 p. 77]{BCRV22}). These Plücker relations are
those described by~\eqref{item:2} in the straightening law from theorem~\ref{thm:straightening}. For any fixed degree $d_c\ge 0$, we can view $\fq[\Cm_I]_{d_c}$ as a free $\fq$-module of rank the number of standard bitableaux of degree $d_c$.
According to \cite[Corollary 3.2.6]{BCRV22}, we have
\begin{align}
  \dim(\fq[\Cm_I])=r(n-r)+1.\label{eq:dimplucker}
\end{align}
More generally, for any polynomial ring $\fq[\Ym]$ in some unknowns $\Ym$, we will denote by $\fq[\Ym][\Cm_I]_{d_c}$ (or $\fq[\Ym]_{d_c}$ for short if it is clear from the context) the $\fq[\Ym]$-module generated by the set of polynomials in $\fq[\Ym,\Cm_I]$ of degree exactly $d_c$ in the $\Cm_I$'s.

Thanks to the straightening law and Proposition~\ref{prop:stab}, we have:
\begin{enumerate}
\item $\fq[\Ym,\Cm_I] = \bigoplus_{d_c \geq 0} \fq[\Ym]_{d_c}$,
\item for all $d_c \geq 1$, $\fq[\Ym]_{d_c}$ is a free
  $\fq[\Ym]$-module of rank the number of standard monomials of $\Cm$
  of shape $(d_c, \dots, d_c)$ ($r$ times):
  \begin{align}
    \rk(\fq[\Ym]_{d_c}) = \det
  \left(  \begin{bmatrix}
      n-j\\
      d_c+j-i		
    \end{bmatrix}_{1 \leq i, j \leq r}\right).\label{eq:rkCm}
  \end{align}
\end{enumerate}
	
\section{Hilbert series of  determinantal Support Minors ideals}
\label{sec:SMHSdet}
Let $\Cm$, $\Um$ be two matrices of variables of size $r\times n$ and $m\times n$.
The goal of this section is to compute the Hilbert series for the
 \emph{determinantal Support Minors system}, i.e. the set of maximal minors of
the matrix $\binom{\Cm}{\Um}$ 
that contains the $r$ rows of $ \Cm$, as an ideal in the  Plücker subalgebra. To this end, we describe a
$\fq$-basis of the ideal in terms of standard monomials, and derive a
first formula~\eqref{eq:HS1} for the Hilbert series. Then, by applying several
formulae from combinatorics, we simplify the formula to get
\eqref{eq:HSB} in theorem~\ref{thm:HSdelta} and~\eqref{eq:HSdcU}.
\subsection{A \texorpdfstring{$\fq$}{K}-basis with standard monomials}
\label{sec:computationHS}
Let
\begin{align}
  \Seq \eqdef \left\{ (i_p,\dots,i_1,r,\dots,1 | b_1, \dots, b_{p+r})
  : \substack{p \geq 1,\; r+1 \leq i_1 < \dots < i_p \leq r+m, \; 1 \leq j_1 < \dots <
  j_{p+r} \leq n} \right\}\label{eq:defSeq}
\end{align}
be the set of bivectors corresponding to minors of
$\binom{\Cm}{\Um}$ that contain all the rows of $\Cm$ and at least one row of
$\Um$ ($p\ge 1$). We write $\idseq = \langle \Seq\rangle$ the ideal of
$\fq[\Um, \Cm]$ generated by $\Seq$.  Then $\idseq$ is exactly
the ideal generated by the support-minors equations in $\fq[\Um, \Cm]$
(without the change for the Plücker coordinates). The following proposition shows that we can see $\idseq$ as a $\fq$-vector space generated by standard monomials.
\begin{proposition}[{\cite[Proposition 3.4.1 p. 83]{BCRV22}}]\label{prop:basis of SM}
  The set $\mathcal Y_{eq}$ of all standard monomials $Y = \gamma_1\dots\gamma_t$ with $\gamma_1 \in \Seq$ form a basis of $\idseq$ as a $\fq$-vector space. 
\end{proposition}
Equivalently, the set of all standard monomials $Y=\gamma_1\dots\gamma_t$ with $\gamma_1 \notin \Seq$ form a basis of $\fq[\Um, \mathbf{C}]/\idseq$ as a $\fq$-vector space.
\begin{proof}
  We give the proof for the sake of completeness. 
  First, note that if $(a|b) \leq (\alpha | \beta)$ and $(\alpha|\beta)\in \Seq$ then $(a|b)\in \Seq$.
Indeed, with the notation $(a|b) = (a_{s'}, \dots, a_1  | b_1, \dots, b_{s'})$ and $(\alpha | \beta) = (\alpha_s, \dots, \alpha_1 | \beta_1, \dots, \beta_s) \in \Seq$,  we must have $s' \geq s \geq r+1$ and $i \leq a_i \leq \alpha_i=i$ for all $1 \leq i \leq r$. Then $(a_1, \dots, a_r) = (1, \dots, r)$ and $(a|b) \in \Seq$.

Clearly $\mathcal Y_{eq}\subset \idseq$.
The straightening law shows that the elements in $\mathcal Y_{eq}$  are linearly independent over $\fq$, and that any element in $\fq[\Um,\Cm]$ is a sum of standard monomials. Any element in $\idseq$ being then a linear combination of elements $\delta Y$ with $Y$ a standard monomial and $\delta \in \Seq$, to conclude the proof we just have to prove that any such element $\delta Y$ is a $\fq$-linear combination of elements in $\mathcal Y_{eq}$.
We can write $\delta Y$ in the basis of the standard monomials $\delta Y = \sum_{i} z_i Y_i$ such that, for all $i$, $z_i\in\fq$ and $Y_i = \gamma_{i, 1} \dots \gamma_{i, t_i}$ with $\gamma_{i, 1} \leq \delta$ according to the point $\eqref{item:4}$ of the straightening law. This implies that $\gamma_{i,1}\in \Seq$ and $Y_i \in \mathcal Y_{eq}$.
\end{proof}	
As the equations of the Support-Minors modeling are polynomials in $\Um$ and the maximal minors of $\Cm$, we would like to study the ideal generated by $\Seq$ not in $\fq[\Um, \Cm]$, but in $\fq[\Um,\Cm_I]$ the algebra generated by $\Um$ and the maximal minors of $\Cm$.
\begin{lemma} \label{lemma:Palgebre}
  The standard bitableaux
  \begin{align*}
    \ytableausetup{mathmode, boxframe=normal, boxsize=2em}
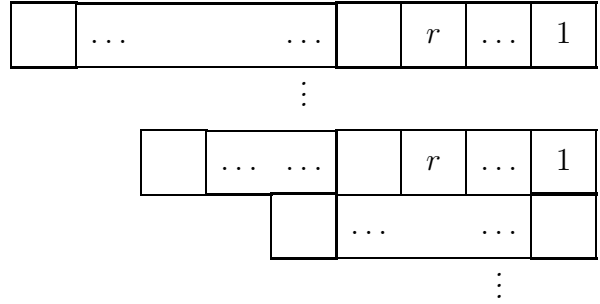
    \begin{ytableau}
      \cline{2-5}
      & \none[\dots] & \none & \none & \none[\dots] & & r & \dots & 1\\
      \cline{2-5}
      \none & \none & \none & \none & \none[\vdots]\\
      \cline{4-5}
      \none & \none & & \none[\dots] & \none[\dots] & & r & \dots & 1\\
      \cline{4-4}
      \none & \none & \none & \none & & \none[\dots] & \none & \none[\dots] & \\
      \cline{6-8}
      \none & \none & \none & \none & \none & \none & \none & \none[\vdots]\\
    \end{ytableau}
  \end{align*}
  where all the coefficients, except the ones in the top-right, are in $\{r+1, \dots, m+r \}$, form a basis of the algebra $\fq[\Um,\Cm_I]$ as a $\fq$-vector space.  As noted in remark~\ref{remark:li}, the result is still true if we replace $\fq$ by a commutative ring $R$, in which case we have a basis as a $R$-module. \end{lemma}
\begin{proof}
  We call $\SM$-monomials (or
  equivalently $\SM$-bitableaux) the monomials such that any two consecutive minors $\gamma_{1}\gamma_{2}$
  with $\gamma_i=(a_i|b_i)$ satisfy one of the following conditions
  (where $\gamma_2$ may be $(|)=1$):
\begin{enumerate}
\item both $a_1$ and $a_{2}$ contain $\{1,\dots,r\}$,
\item or $a_1$ contains $\{1,\dots,r\}$ and $a_{2}$ involves no rows from $\{1,\dots,r\}$,
\item or both $a_1$ and $a_{2}$ involve no rows from $\{1,\dots,r\}$.
\end{enumerate}
The standard bitableaux defined in the lemma
are exactly standard $\SM$-bitableau,     and belong to $\fq[\Um,\Cm_I]$. As noted in remark~\ref{remark:li}, the standard bitableaux are
    always linearly independent over any commutative ring.  To conclude the proof, it is
    sufficient to show that any monomial $M=M_{\Cm_I}M_{\Um}$ in
    $\fq[\Um,\Cm_I]$, with $M_{\Cm_I}$ a monomial in $\Cm_I$ and
    $M_{\Um}$ a monomial in $\Um$, can be written as a sum of standard
    $\SM$-bitableaux.

The minors appearing in $M_{\Um}$ can be seen as minors in $\Um,\Cm_I$
by adding $r$ on all the coefficients. Now all the minors that appear
in this decomposition of $M=M_{\Cm_I}M_{\Um}$ have the shape $(a|b)$
where $a$ either involves all the rows from $\{1,\dots,r\}$ (the minor
divides $M_{\Cm_I}$), or none of them (the minor divides $M_{\Um}$),
hence $M$ is a $\SM$-bitableau.

We know by~\eqref{item:3} from the straightening law that any
non-standard bitableau can be expressed as a sum of standard
bitableaux by applying~\eqref{item:2} several times. To conclude the
proof, we show that each time we rewrite a non-standard
product $\gamma_1 \gamma_{2}$  in a $\SM$-monomial, the resulting monomials are still
$\SM$-monomials. For any minor $(a|b)$ in a $\SM$-monomial, if $a$ contains
$\{1,\dots,r\}$ (resp. no row from $\{1,\dots,r\}$), then all
previous (resp. following) minors do too. It remains to show that
after rewriting some $\gamma_1\gamma_2$, we get a sum of standard linearly independent terms
$\epsilon\eta$ that stay in the same case $(1)$-$(3)$ as
$\gamma_1\gamma_2$. The law tells us that $\epsilon <
\gamma_1$. Remember that all bitableaux correspond to homogeneous
polynomials, and the degree in $\Cm$ of a monomial is the number of
occurrences of the values $\{1,\dots,r\}$ in the bitableaux.  For case
$(1)$, as $\Cm$ only contains $r$ rows, necessarily also both $\eta$
and $\epsilon$ must involve the rows $\{1,\dots,r\}$.  For case $(2)$,
$\epsilon<\gamma_1$ implies that $\epsilon$ involves the rows
$\{1,\dots,r\}$, and by a degree argument $\eta$ involves no rows from
$\{1,\dots,r\}$. For case $(3)$, again by a degree argument, both
$\epsilon$ and $\eta$ involve no rows from $\{1,\dots,r\}$.
\end{proof}
This implies that, for any $d_c\ge 1$,  $\fq[\Um]_{d_c}$ is the module generated by all these standard monomials such that exactly the $d_c$ first rows have a left tableau of the form:
\begin{align*}
  \ytableausetup{mathmode, boxframe=normal, boxsize=2em}
  \begin{ytableau}
    \cline{2-5}
    & \none[\dots] & \none & \none & \none[\dots] & & r & \dots & 1\\
    \cline{2-5}
  \end{ytableau}
\end{align*}	
For all $d_c \geq 1$, we note\footnote{$\mathcal S_{d_c}$ can also be defined as the intersection of  the ideal generated by $\mathcal S_{eq}$ in $\fq[\Um,\Cm_I]$ with $\fq[\Um]_{d_c}$, thanks to Lemma~\ref{lemma:Palgebre} and Proposition~\ref{prop:basis of SM}.} $\mathcal{S}_{d_c} = \fq[\Um]_{d_c} \cap \idseq$, which is a graded submodule of $\fq[\Um]_{d_c}$, and we want to compute the Hilbert series of $\fq[\Um]_{d_c}/ \mathcal{S}_{d_c}$, which is the $\fq$-vector space generated by the standard monomials whose left-hand tableau is of the form:	
\begin{align*}
  \ytableausetup{mathmode, boxframe=normal, boxsize=2em}
  \begin{ytableau}
    \cline{2-4}
    r & \none & \none[\dots] & \none & 1 \\
    \cline{2-4}
    \none & \none & \none[\dots] & \none & \none \\
    \cline{2-4}
    r & \none & \none[\dots] & \none & 1 \\
    \cline{2-4}
    & \none & \none[\dots] & \none & \\
    \cline{2-4}
    \none & \none & \none & \none[\vdots] \\
    \cline{4-4}
    \none & \none & & \none[\dots] & \\
    \cline{4-5}
  \end{ytableau}
\end{align*}
We can enumerate these standard bitableaux by counting the lower part of them, each row being of length smaller than $r$ with coefficients in $\{r+1,\dots,r+m\}$. Moreover, if we fix the shape of a bitableau, we can enumerate the left part and the right part (whose coefficients are in $\{1, \dots, n\}$) independently and have, for a degree in $\Um$ fixed to $d_u$: 
\begin{align}
  \dim_\fq(\fq[\Um]_{d_c} / \mathcal{S}_{d_c})_{d_u} = \sum_{v \overset{r}{\rightsquigarrow} d_u} \stab(m, v(1), \dots, v(r)) \cdot \stab(n, v(1)+d_c, \dots, v(r)+d_c)\label{eq:HSstab}
\end{align}
With the  explicit formula for $\stab$  given in
Proposition~\ref{prop:stab}, we obtain the first explicit formula
for the Hilbert series of $\fq[\Um]_{d_c} / \mathcal{S}_{d_c}$:
\begin{theorem}\label{thm:HS1}
  For all $d_c \geq 1$, we have:
  \begin{align}
    \HS_{\fq[\Um]_{d_c} / \mathcal{S}_{d_c}}(t) = \sum_{d_u \geq 0} m_{d_u, d_c} t^{d_u}\label{eq:HS1}
  \end{align}
  where 
  \begin{align*}
    m_{d_u, d_c} = \sum_{v \overset{r}{\rightsquigarrow} d_u} \det
\left(    \begin{bmatrix}
      m-j\\
      v(i)+j-i
\end{bmatrix}
    \right)_{i,j}
    \det \left(
    \begin{bmatrix}
      n-j\\
      v(i)+d_c+j-i
    \end{bmatrix}
    \right)_{i,j}
  \end{align*}
  where the sum ranges over the tuples $v=(v(1),\dots,v(r))$ such that $\sum_{i=1}^r v(i)=d_u$, $d_u\ge v(1)\ge \dots \ge v(r)\ge 0$ and the indices of the matrices are $i,j\in\{1,\dots,r\}$.
\end{theorem}
\subsection{Factorization of the Hilbert series}
The explicit formula from theorem~\ref{thm:HS1} is not easy to compute, or to compare with existing formulae for other Hilbert Series (for instance for a regular system, or for the Minors system). We use combinatorial results to rewrite this series as a determinant of a matrix (theorem~\ref{thm:HSdelta}), and then show that this determinant can be factorized for a very simple formula (theorem~\ref{thm:HSA}).
\begin{theorem}\label{thm:HSdelta}
  For all $d_c \geq 1$, 
  \begin{align}\label{eq:HSdelta}
    \HS_{\fq[\Um]_{d_c} / \mathcal{S}_{d_c}}(t) = \det \left( \Delta_{d_c}(t) \right)
  \end{align}  
  where $\Delta_{d_c}(t) = \left( \sum\limits_{\ell \geq 0}
    \begin{bmatrix}
      m-i\\
      \ell
    \end{bmatrix} 
    \begin{bmatrix}
      n-j\\
      \ell+d_c+j-i
    \end{bmatrix} t^\ell \right)_{1 \leq i, j \leq r}$.
\end{theorem}
\begin{proof}
  We adapt and simplify the proof from~\citet[p.15]{G83} to our context. The proof is quite similar, but with the introduction of some $d_c$'s.
Let us define the $r\times (d_u+r)$ matrices
  \begin{align*}
    E_{d_u} &= \begin{pmatrix}
                            \begin{bmatrix}
                            	m-i \\
                            	\ell-(r-i)
                            \end{bmatrix} t^{\ell-(r-i)}
    \end{pmatrix}_{1 \leq i \leq r, \; 0 \leq \ell \leq d_u+r-1}                
    \\
    F_{d_u} &= \begin{pmatrix}
                             \begin{bmatrix}
                            	n-j \\
                            	\ell+d_c-(r-j)
                            \end{bmatrix}
                           \end{pmatrix}_{1 \leq j \leq r,\; 0 \leq \ell \leq d_u+r-1}
  \end{align*}
  From~\eqref{eq:stab}, for any $r$-uple $v$ with $0\le v(r)\le \dots\le v(1)\le d_u$, the maximal minor of $E_{d_u}$ and $F_{d_u}$ defined by the columns $v(1)+r-1, v(2)+r-2, \dots, v(r)$ are respectively $\stab(m, v(1), \dots, v(r)) t^{\sum_i v(i)}$ and $\stab(n, v(1)+d_c, \dots, v(r)+d_c)$. Moreover, we have a bijection between  subsets $\{i_1<\dots<i_r\}\subset \{0,\dots,d_u+r-1\}$  of size $r$ and $r$-uples $v$ with $0\le v(r)\le \dots\le v(1)\le d_u$ given by $v(j)+r-j=i_j$ for $j\in\{1,\dots,r\}$.
  By using the Cauchy-Binet formula for the determinant of a product of non-square matrices, we obtain:  
  \begin{align*}
   \det \left( E_{d_u} F_{d_u}^T \right) =
    \sum_{\substack{0 \leq v(r) \leq \dots \leq v(1) \leq d_u}} \stab(m, v(1), \dots, v(r)) \cdot \stab(n, v(1)+d_c, \dots, v(r)+d_c)\cdot t^{\sum_i v(i)}.
  \end{align*}
 Moreover, if we compute directly the matrix product, we get with a change of variable:
  \begin{align}
   \det \left( E_{d_u} F_{d_u}^T \right)  =&\det \left( \sum\limits_{\ell = 0}^{d_u+r-1} 
      \begin{bmatrix}
        m-i\\
        \ell-r+i
      \end{bmatrix}
      \begin{bmatrix}
        n-j\\
        \ell+d_c-r+j
      \end{bmatrix}
      t^{\ell-r+i} \right)_{1\le i,j\le r} \notag{}\\
=&\det \left( \sum\limits_{\ell = 0}^{d_u+i-1} 
      \begin{bmatrix}
        m-i\\
        \ell
      \end{bmatrix}
      \begin{bmatrix}
        n-j\\
        \ell+d_c-i+j
      \end{bmatrix}
      t^{\ell} \right)_{1\le i,j\le r}\label{eq:Dt1}.
  \end{align}
  According to Equation~\eqref{eq:HSstab},
  the coefficient of degree $d_u$ of the Hilbert series
  is the coefficient of degree $d_u$ in $  \det \left( E_{d_u} F_{d_u}^T \right)$, as the sum ranges over the shapes $(v(1),\dots,v(r))$ of degree $d_u$.
  Thanks to the previous equality, we can equivalently take the coefficient of degree $d_u$ in the determinant~\eqref{eq:Dt1}.

  As all entries of the matrix are polynomials in $t$ of degree  $\ge d_u$ (the entries are twisted coefficients that are always non zero), we can add terms of larger degree in the matrix without changing the value of the coefficient of degree $d_u$, and consider the determinant of a matrix of formal power series in $t$ that does not depend on $d_u$. This concludes the  proof of the theorem.
\end{proof}
We now use a  Saalschütz formula to factorize the coefficients of $\Delta_{d_c}(t)$.
\begin{lemma}[Saalschütz formula \cite{GS85}]
  For all $\ell, f$ non negative integers, and $a$ and $b$
  arbitrary numbers, we have
  \begin{align}
    \sum_{k\ge 0} \binom{b}{f-k}\binom{a}{\ell-k}\binom{a+b+k}{k} = \binom{a+f}{\ell}\binom{b+\ell}{f}.\label{eq:saalschütz}
  \end{align}
  Note that the equality remains true for $f< 0$ if we take $\binom{n}{k}=0$ for $k<0$ by convention.
\end{lemma}
We deduce the following compact form for the Hilbert series.
\begin{theorem}
	For all $d_c \geq 1$, 
	\begin{align}\label{eq:HSB}
    \HS_{\fq[\Um]_{d_c} / \mathcal{S}_{d_c}}(t) = \frac{\det(B_{d_c}(t))}{(1-t)^{(m+n-r)r}}
  \end{align}
  where $B_{d_c}(t) = \left( \sum_{\ell \geq 0} \binom{n+d_c-i}{\ell+d_c+j-i} \binom{m-d_c-j}{\ell} t^\ell \right)_{1 \leq i, j \leq r}$
\end{theorem}
\begin{proof}
  First, we show that we can factorize $1/(1-t)^{m+n-i-j+1}$ in each coefficient $(i,j)$ of the matrix $\Delta_{d_c}(t)$.
  Let $G = (m-i) + (n-j) + 1$, we want to show that
  \begin{align*}
    \Delta_{d_c, i, j}(t) &= \frac{1}{\left(1-t\right)^{G}} \sum_{\ell \geq 0} \binom{n+d_c-i}{\ell+d_c+j-i} \binom{m-d_c-j}{\ell} t^\ell.
  \end{align*}
	Let us call $\tilde{\Delta}(t)$ the right hand part of the equality.
Expanding $1/(1-t)^G=\sum_{u\ge 0}\binom{G-1+u}{u}t^u$ in power series, and collecting the terms by powers of $t$, we get up to a relabeling of the summation indexes:
\begin{align*}
  \tilde{\Delta}(t) &= \sum_{\ell\ge 0}\left[\sum_{k\ge 0}  \binom{n+d_c-i}{\ell+d_c+j-i-k} \binom{m-d_c-j}{\ell-k} \binom{G-1+k}{k}\right] t^\ell
\end{align*}
As $G-1=m+n-i-j = (n+d_c-i)+(m-d_c-j)$, we can apply~\eqref{eq:saalschütz} to get the wanted equality:
\begin{align*}
  \tilde{\Delta}(t) &= \sum_{\ell\ge 0} \binom{m-d_c-j+\ell+d_c+j-i}{\ell}\binom{n+d_c-i+\ell}{\ell+d_c+j-i}t^\ell\\
  &= \sum_{\ell\ge 0}
    \begin{bmatrix}
      m-i\\
      \ell
    \end{bmatrix}
    \begin{bmatrix}
n-j\\\ell+d_c+j-i
    \end{bmatrix}
    t^\ell = \Delta_{d_c,i,j}(t).
\end{align*}
This proves that each coefficient in row $i$ and column $j$ of $\Delta_{d_c}(t)$ has $1/(1-t)^{m+n-i-j+1}$ in factor, for $1\le i,j\le r$. Then,
each  row $i$  has $1/(1-t)^{m-i+1}$ in factor, so that we can factorize the determinant by $1/(1-t)^{\sum_{i=1}^r m-i+1}$. After that, each coefficient in row $i$ and column $j$ has $1/(1-t)^{n-j}$ in factor, so that we can again factorize $1/(1-t)^{n-j}$ in each column $j$ to get
\begin{align*}
  \det(\Delta_{d_c}(t)) = \frac{1}{(1-t)^N}\det \left( \sum_{\ell \geq 0} \binom{n+d_c-i}{\ell+d_c+j-i} \binom{m-d_c-j}{\ell} t^\ell \right)_{1 \leq 1, j \leq r}
\end{align*}
with $N = \sum_{i=1}^r (m-i+1) + \sum_{j=1}^r (n-j) = r(m+n-r)$.
\end{proof}
Finally, we can apply the following lemma  to transform once again the expression of the Hilbert series.
\begin{lemma}[{\cite[p. 679]{CH94}}]
Let us consider the following $r \times r$ matrices of polynomials, for any $m,n,d_c\in\mathbb N$, $1\le r \le \min(m,n)$ and $i,j\in\{1,\dots,r\}$:
  \begin{align*}
    B_{d_c} &= \left( \sum\limits_{\ell \ge 0} \binom{m-d_c-j}{\ell} \binom{n+d_c-i}{\ell+d_c+j-i} t^\ell \right)_{1\le i,j\le r},\\
    B_{d_c}' &= \left( \frac{(-1)^{i+j}}{t^{j-1}} \sum\limits_{\ell \ge 0} \binom{m-d_c-j}{\ell} \binom{n+d_c-i}{\ell+d_c} t^\ell \right)_{1\le i,j\le r},
  \end{align*}
  \begin{align*}
    T = \left( (-1)^{j-i} \binom{j-1}{i-1} \frac{1}{t^{j-i}} \right)_{1\le i,j\le r}, \text{ and } \, T' = \left( (-1)^{i-j} \binom{i-1}{j-1} \right)_{1\le i,j\le r}.
  \end{align*}
  Then,  $B_{d_c}' = T' B_{d_c} T$. In particular, since $T'$ and $T$ are triangular matrices whose diagonal elements are $1$, it follows that $\det(B_{d_c}) = \det(B_{d_c}')$.
\end{lemma}
\begin{proof}
  We give the proof for the sake of completeness. The result comes
  from the use twice of the following formula, that is given
  in~\cite{CH94} without proof.
  \begin{align}
    \sum_{k\ge 1} (-1)^{i-k}\binom{i-1}{k-1}\binom{a-k}{b-k} = (-1)^{i-1}\binom{a-i}{b-1} \text{ with } i\ge 1, a \in \mathbb C, b \ge 1.\label{eq:combi}
  \end{align}
 We did not find another reference for this formula, so that we give an elementary proof of it in~\ref{sec:appendix}. Now using this formula with
 \begin{enumerate}
 \item $a=n+d_c, b=\ell+d_c+j$, as $i\le r$\label{item:i}
 \item \label{item:ii} $a=m-d_c, \; b=\ell+1, \text{ as } j\le r.$
 \end{enumerate}
we get
  \begin{align*}
    (T'B_{d_c})_{i,j}  &= \sum_{\ell\ge 0} \sum_{k=1}^r (-1)^{i-k}\binom{i-1}{k-1}\binom{m-d_c-j}{\ell}\binom{n+d_c-k}{\ell+d_c+j-k}t^\ell\\
                 &\overset{\eqref{item:i}}{=}\sum_{\ell\ge 0}\binom{m-d_c-j}{\ell}(-1)^{i-1}\binom{n+d_c-i}{\ell+d_c+j-1}t^\ell\\
    (T'B_{d_c}T)_{i,j} &= \sum_{\ell\ge 0}\sum_{k=1}^r\binom{m-d_c-k}{\ell}(-1)^{i-1}\binom{n+d_c-i}{\ell+d_c+k-1}(-1)^{j-k}\binom{j-1}{k-1}\frac{t^\ell}{t^{j-k}} \\
                 &=\frac{(-1)^{i+j}}{t^{j-1}} \sum_{\ell\ge 0}\sum_{k=1}^r(-1)^{k-1} \binom{m-d_c-k}{\ell-k+1}\binom{n+d_c-i}{d_c+\ell}\binom{j-1}{k-1}t^{\ell}\\
    &\overset{\eqref{item:ii}}{=} \frac{(-1)^{i+j}}{t^{j-1}}\sum_{\ell\ge 0} \binom{n+d_c-i}{d_c+\ell}\binom{m-d_c-j}{\ell}t^{\ell}=(B'_{d_c})_{i,j}.
  \end{align*}
\end{proof}
We deduce the following compact form for the Hilbert series.
\begin{theorem}\label{thm:HSA}
  For any $d_c\ge 1$, we have:
  \begin{align}\label{eq:HSA}
    \HS_{\fq[\Um]_{d_c} / \mathcal{S}_{d_c}}(t) = \frac{\det(A_{d_c}(t))}{t^{\binom{r}{2}} (1-t)^{(m+n-r)r}}
  \end{align}
  where $A_{d_c}(t) = \left( \sum\limits_{\ell \geq 0} \binom{n+d_c-i}{\ell+d_c} \binom{m-d_c-j}{\ell} t^\ell \right)_{1 \leq i, j \leq r}$.
\end{theorem}
\begin{proof}
  The coefficients of $A_{d_c}$ are the same as the ones for $B_{d_c}'$ up to a factor $\frac{(-1)^{i+j}}{t^{j-1}}$ for entry $(i,j)$. We can factorize $(-1)^i$ on each row $i$, and independently $\frac{(-1)^j}{t^{j-1}}$ on each column $j$. Then
  \begin{align*}
    \det(B_{d_c})=    \det(B_{d_c}') =(-1)^{\sum_{i=1}^r i + \sum_{j=1}^r j}\prod_{j=1}^r\frac1{t^{j-1}}\det(A_{d_c}) = \frac1{t^{\binom{r}{2}}}\det(A_{d_c}).
  \end{align*}
\end{proof}

\section{Support-Minors Hilbert series in the generic case (\texorpdfstring{$K\ge m(n-r)$}{K>=m(n-r)})}
\label{sec:SMHSgeneral}
Following~\citet{FSS13}, we transfer properties for determinantal
$\SM$ ideals to ideals corresponding to the $\SM$
modeling.

Let $\mathcal D$ be the ideal of $\fq[\Um]$ generated by all the $(r+1) \times (r+1)$ minors of $\Um$ (the classical generic minors ideal). We denote by $\mathcal S$ the ideal of $\fq[\Um, \Cm_I]$ generated by $\mathcal D$ and $\mathcal S_{eq}$.
We start by proving in section~\ref{sec:cohenmacaulay} that $B[\Um, \Cm_I] / \mathcal S$ is Cohen-Macaulay for any Cohen-Macaulay ring $B$, and deduce its Krull dimension. Note that we need here to include $\mathcal D$ in $\mathcal S$, otherwise we do not get a Cohen-Macaulay ring.

Then with almost the same proof as~\citet{FSS13}, adapted to our modules, we show in section~\ref{sec:mngeneric} that for $K\ge m(n-r)$, we can add to $\mathcal S_{d_c}$  $mn$ generic\footnote{i.e. sums of all the possible monomials of
  degree $D$ (with a weight $D$ on the variables $\Um$) where the
  coefficients are new variables.} forms
  $g_{i,j}$ of degree $D$ that are non-zero divisor in $\mathcal S_{d_c}$, and compute the Hilbert series of the resulting module.
  We show in section~\ref{sec:uij}
  that the Hilbert series is the same if we add polynomials
  $u_{i,j}-f_{i,j}$ with $f_{i,j}$ generic forms of degree $D$ in the
  $\xv$ variables instead of the $g_{i,j}$'s. This is done by showing that, in both cases (the
  $g_{i,j}$'s or the $u_{i,j}-f_{i,j}$'s), there exists non-empty
  Zariski open sets on which the specialized ideal has the same
  Hilbert series as the generic one.
  
\subsection{Some properties of \texorpdfstring{$\mathcal S$}{S}} 
\label{sec:cohenmacaulay}
We show that adding $\mathcal D$ to $\mathcal S_{eq}$ does not change the modules $\mathcal S_{d_c}$ for $d_c\ge 1$, and that $\mathcal S_{d_c}$, that is defined w.r.t. $\mathcal S_{eq}$ and not $\mathcal S$, is indeed the homogeneous component of $\mathcal S$ of degree $d_c$.
\begin{lemma}\label{lemma:DHinS}
	Let $d_c \geq 1$, then $\mathcal D \fq[\Um]_{d_c} \subset \mathcal S_{d_c}$, and $\mathcal S\cap \fq[\Um]_{d_c}=\mathcal S_{d_c}$.
\end{lemma}

\begin{proof}
	Let $\mathcal D_{\Cm,\Um}$ be the ideal generated by all the $(r+1) \times (r+1)$ minors of
	$\binom{\Cm}{\Um}$. This ideal is generated by all the standard monomials in $\Cm, \Um$ of length greater than or equal to $r+1$ \cite[Corollary 3.4.2 p.83]{BCRV22}. It follows from Lemma~\ref{lemma:Palgebre} and Proposition~\ref{prop:basis of SM} that $\mathcal D_{\Cm,\Um} \cap \fq[\Um]_{d_c} = \mathcal S_{d_c}$. Moreover, as $\mathcal D\subset\mathcal D_{\Cm,\Um}$, for any maximal minor $\Cm_T$   of $\Cm$, we have $\mathcal D \mathcal \Cm_T \subset \mathcal D_{\Cm,\Um} \cap \fq[\Um]_{1}$, hence $\mathcal D\fq[\Um]_{d_c}\subset\mathcal S_{d_c}$. We conclude the proof by using the fact that $\mathcal S = \langle\mathcal S_{eq}\rangle_{\fq[\Um,\Cm_I]} + \mathcal D \fq[\Um,\Cm_I]$. 
If we consider the graded component of degree $d_c\ge 1$, as the ideals are homogeneous, we have $\mathcal S\cap\fq[\Um]_{d_c} = (\langle\mathcal S_{eq}\rangle_{\fq[\Um,\Cm_I]})_{d_c} + (\mathcal D \fq[\Um,\Cm_I])_{d_c} = \mathcal S_{d_c} + \mathcal D \fq[\Um]_{d_c} = \mathcal S_{d_c}.$
\end{proof}
This is interesting because we aim to view the quotient by $\mathcal S$ as an algebra with straightening law (ASL) 
in order to show it is a Cohen-Macaulay ring. The ideal generated by $\Seq$ alone is not Cohen-Macaulay, unlike the ideal generated by both $\Seq$ and $\mathcal D$, as explained below.

\begin{corollary}
  Let $\mathcal J$ be the Minors ideal, i.e. the ideal of $\fq[\xv]$ generated by the $(r+1)\times (r+1)$ minors of $\Fm$. Let $C_T$ be any maximal minor of $\Cm$, and  $\idealSM$ be the Support Minors ideal in $\fq[\xv][\Cm_I]$. Then 
  \begin{align*}
    \mathcal{J} C_T \subset \idealSM
  \end{align*}
  where $\mathcal JC_T$ is the set of all elements in $\mathcal J$ multiplied by the minor $C_T$. 
\end{corollary}
\begin{proof}
  By elimination of the $u_{i,j}$ variables (that do not appear in $\mathcal D \Cm_T$ nor in the $f_{i,j}$'s), we have
  \begin{align*}
\mathcal{J}C_T = \langle \mathcal{D}C_T  + \langle u_{i, j} - f_{i, j} \rangle \rangle_{\fq[\xv,\Um,\Cm_I]} \cap \fq[\xv,\Cm_I].
  \end{align*}
  In the other hand, $\idealSM$ is the ideal generated by the Support Minors equations for $\Fm$, so that 
  \begin{align*}
    \idealSM = \left(\langle \Seq\rangle_{\fq[\xv,\Um,\Cm_I]} + \langle  u_{i, j} - f_{i, j} \rangle_{\fq[\xv,\Um,\Cm_I]}\right) \cap \fq[\xv,\Cm_I].
  \end{align*}
  According to lemma~\ref{lemma:DHinS}, for any fixed minor $C_T$, we have
  \begin{align*}
    \langle \mathcal D \Cm_T \rangle_{\fq[\Um,\Cm_I]}  \subset \oplus_{d_c\ge 1}  \mathcal D {\fq[\Um]_{d_c}}
\subset \oplus_{d_c\ge 1} \mathcal S_{d_c} = \langle \mathcal S_{eq}\rangle_{\fq[\Um,\Cm_I]}
  \end{align*}
  Adding the $\xv$ variables and the polynomials $u_{i,j}-f_{i,j}$ we get
  \begin{align*}
    \langle \mathcal D \Cm_T \rangle_{\fq[\xv,\Um,\Cm_I]} +  \langle u_{i, j} - f_{i, j} \rangle \rangle_{\fq[\xv,\Um,\Cm_I]}
    \subset
     \langle \mathcal S_{eq}\rangle_{\fq[\xv, \Um,\Cm_I]} +  \langle u_{i, j} - f_{i, j} \rangle \rangle_{\fq[\xv,\Um,\Cm_I]}
  \end{align*}
  By eliminating the $\Um$ variables, we get $\mathcal J \Cm_T\subset \idealSM$.
\end{proof}
\subsection{Cohen-Macaulay properties}
We use known results about algebras with straightening law to show that  $\fq[\Um, \Cm_I] / \mathcal S$ is a Cohen-Macaulay ring  and compute its dimension.
 \begin{definition}[{\cite[p.58]{BV06}}]
  Let $\Pi$ be a partially ordered set.
  \begin{enumerate} 
  \item  An element $\beta \in \Pi$ is a \emph{cover} of $\alpha \in \Pi$ if $\beta > \alpha$ and if there is no element $\gamma \in \Pi$ such that $\beta > \gamma > \alpha$.
  \item $\Pi$ is said to be \emph{wonderful} if the following holds after a smallest and a greatest element $-\infty$ and $\infty$ have been added to $\Pi$ : if $\alpha \in \Pi \cup \{-\infty\}, \gamma \in \Pi \cup \{\infty\}$, and $\beta_1, \beta_2 \in \Pi$ are covers of $\alpha$ satisfying $\beta_1, \beta_2 < \gamma$, then there exists an element $\beta \in \Pi \cup \{\infty\}$, with $\beta \leq \gamma$ which covers both $\beta_1$ and $\beta_2$.
  \end{enumerate}
\end{definition}
Graphically, it can be pictured as in figure~\ref{fig:wonderful}.
\begin{figure}[h]
  \centering
  \begin{tikzpicture}
    \node (g) at (0,-1) {$\gamma$};
    \node (b) at (0,0) {$\beta$};
    \draw [dashed] (g) -- (b);
    \node (b1) at (-1,0.5) {$\beta_1$};
    \draw [dashed] (g) -- (b1);
    \node (b2) at (1,0.5) {$\beta_2$};
    \draw [dashed] (g) -- (b2);
    \node (a) at (0,1) {$\alpha$};
    \draw (b) -- (b1) -- (a) -- (b2) -- (b);
    \node (small) at (2,1) {};
    \node (large) at (2,-1) {};
    \draw[-stealth] (small)--(large);
  \end{tikzpicture}
  \caption{Graphical image of a wonderful set. A line means that the top element is smaller than the  one below. A solid line means that the  element below is a cover for the one above.}
  \label{fig:wonderful}
\end{figure}
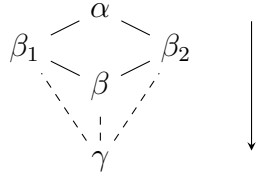

\begin{theorem}{{\cite[theorem 5.14]{BV06}}}\label{thm:wonderful}
  Any graded ASL on a wonderful set $\Pi$ over a Cohen-Macaulay ring $B$ is also Cohen-Macaulay.
\end{theorem}
 Let us define $\piS = \piC \cup \piU$ where
 \begin{align*}
   \piC &= \left\{ (r, \dots, 1 | j_1, \dots, j_r) : 1 \leq j_1 < \dots < j_r \leq n \right\},\\
   \piU &= \left\{ (i_p, \dots, i_1 | j_1, \dots, j_p) : 1 \leq p \leq r, \substack{r+1 \leq i_1 < \dots < i_p \leq m+r, 1 \leq j_1 < \dots < j_r \leq n }\right\}.
 \end{align*}
For any $B$ a Cohen-Macaulay ring,  $B[\Um, \Cm_I] / \mathcal S$ is an ASL on $\piS$, from the Straightening Law~\ref{thm:straightening}, Proposition~\ref{prop:basis of SM} and Lemma~\ref{lemma:Palgebre}. The minors corresponding to $\piC$ are maximal minors of $\Cm$, whereas those corresponding to $\piU$ are minors just in $\Um$ of size at most $r$.
Remark that the fact that we restrict to $p\le r$ in $\piU$ comes from the addition of $\mathcal D$ to $\mathcal S$. Without $\mathcal D$, $\piU$ would have been defined for any $p\in\{1,\dots,\min(m,n)\}$. We draw an example of Hasse diagram for $m=2, n=3, r=2$ in figure~\ref{fig:hasse}.
\begin{lemma}\label{lemma:minmax}
For any minors $\alpha$ and $\beta$  such that $\alpha\le\beta$, $\alpha\in\piU$ implies $\beta\in\piU$. Moreover, $\alpha_{\min}=(r,\dots,1|1,\dots,r)\in\piC$ is smaller than any element in $\piS$, and $\alpha_{\max}=(m+r|n)\in\piU$ is larger than any element in $\piS$.
\end{lemma}
\begin{lemma}\label{lemma:cover}
  Let  $\ev_i$ denote the canonical vector of zeros with only one 1 in position $i$, and $\zerom$ the zero vector.
We denote by $E_R$ the set of bivectors of the form $(\zerom | \ev_i)$ for $i\in\{1,\dots,n\}$, by $E_L$ the sets of bivectors of the form $(\ev_j | \zerom)$ for $j\in\{1,\dots,m+r\}$ and $E = E_L \cup E_R$.
  Let $\alpha=(a_{p},\dots,a_1|a'_1,\dots,a'_p)\in\piS, \beta=(b_s,\dots,b_1|b'_1,\dots,b'_s)\in\piS$ such that $\beta$ is a cover for $\alpha$ in $\piS$, then
  \begin{enumerate}
  \item if both $\alpha,\beta\in\piU$,
    \begin{itemize}
    \item if $p=s$ there exists $\epsilon\in E$  such that  $\beta=\alpha+\epsilon$.
    \item if $p>s$, then $a_p=m+r$, $a'_p=n$ and $\beta=(a_{p-1},\dots,a_1|a'_1,\dots,a'_{p-1})$.
    \end{itemize}
    \item if both $\alpha,\beta\in\piC$,  there exists $\epsilon\in E_R$  such that  $\beta=\alpha+\epsilon$.
    \item if $\alpha\in\piC$, $\beta\in\piU$, then $s=p=r$, $(b_r,\dots,b_1)=(2r,\dots,r+1)$ and $(b'_1,\dots,b'_r)=(a'_1,\dots,a'_r)$.
  \end{enumerate}
\end{lemma}
\begin{proof}
  If $\alpha<\beta$ then $p\ge s$ and $a_i\le b_i$, $a'_i\le b'_i$ for all $i\in\{1,\dots,s\}$, and at least one of the inequalities is strict.
  \begin{enumerate}
  \item if both $\alpha,\beta\in\piU$,
    \begin{itemize}
    \item  if $p=s$, then some $a_i<b_i$, hence all $a_j=b_j$ for $j\ne i$ and $b_i=a_i+1$ (otherwise we can find $\gamma=a+(e_i|0)\in\piU$ such that $\alpha<\gamma<\beta$) and $\beta=\alpha+(e_i|0)$. If  $a'_i<b'_i$, then $\beta=\alpha+(0|e_i)$. 
    \item if $p>s$, then $\gamma=(m+r,a_{p-1},\dots,a_1|a'_1,\dots,a'_{p})\in\piU$ is such that $\alpha\le \gamma < \beta$ which implies $\gamma=\alpha$ by definition of a cover. It implies that $a_p=m+r$ and $a'_p=n$ in the same way. Let $\gamma=(a_{p-1},\dots,a_1|a'_1,\dots,a'_{p-1})\in\piU$, then $\alpha<\gamma\le \beta$ so that $\gamma=\beta$.
    \end{itemize}
  \item it is similar to the first point works, looking only at the right part.
  \item not all minors between $\alpha$ and $\beta$ are in $\piS$. Here, $\gamma=(2r,\dots,r+1|a'_1,\dots,a'_r)\in\piU$ satisfies $\alpha<\gamma\le \beta$ so that $\gamma=\beta$.
  \end{enumerate}
\end{proof}
\begin{proposition}
	$\piS$ is a wonderful set.
\end{proposition}
\begin{proof}
Let $\alpha \in \piS \cup \{-\infty\}$, $\beta_1, \beta_2 \in \piS$ covers of $\alpha$ and $\gamma \in \piS \cup \{\infty\}, \gamma > \beta_1, \beta_2$. We enumerate all possible cases for $\alpha, \beta_1, \beta_2$. Remark that if $\beta_1=\beta_2$, then there exists a cover $\beta$ of $\beta_1$ such that $\beta_1<\beta\le \gamma$.
\begin{enumerate}
\item if $\alpha \in \piU$, then $\gamma, \beta_1, \beta_2 \in \piU$ by lemma~\ref{lemma:minmax}. Let $p$ be the length of $\alpha$. 
  \begin{enumerate}
  \item if $\beta_1$ and $\beta_2$ have length $p$, by lemma~\ref{lemma:cover} there exists $\epsilon_1$ and $\epsilon_2$ in $E$ such that $\beta_1 = \alpha + \epsilon_1$ and $\beta_2 = \alpha + \epsilon_2$. For $\beta_1\ne \beta_2$ then $\epsilon_1 \neq \epsilon_2$ and we take $\beta = \alpha + \epsilon_1 + \epsilon_2\le \gamma$.
  \item if some $\beta_i$ (let's say $\beta_2$) has length $p-1$, then the length of $\gamma$ is $\le p-1$. From lemma~\ref{lemma:cover}, $\beta_2$ is
    obtained by removing the largest coefficients $a_p=m+r$ and $a'_p=n$
    from $\alpha$. If $\beta_1\ne \beta_2$, then $\beta_1$ has length
    $p$ and there exists $\epsilon \in E$ such that
    $\beta_1 = \alpha + \epsilon$. Then, $\beta=\beta_2+\epsilon$ is
    obtained from $\beta_1$ by removing the largest coefficients $m+r$
    and $n$, hence $\beta$ is a cover for both $\beta_i$'s and we have
    $\beta\le \gamma$.
  \end{enumerate}
\item if $\alpha \in \piC$ and $\beta_1, \beta_2 \in \piC$, then by lemma~\ref{lemma:cover} there exists $\epsilon_1$ and $\epsilon_2$ in $E_R$ such that $\beta_1 = \alpha + \epsilon_1$ and $\beta_2 = \alpha + \epsilon_2$. If $\beta_1\ne\beta_2$ then $\epsilon_1 \neq \epsilon_2$ and  $\beta = \alpha + \epsilon_1 + \epsilon_2\le \gamma$ is a cover in $\piC$ for both $\beta_i$'s.
\item if $\alpha=(r,\dots,1|\,\alpha^{(r)}) \in \piC$ and  $\beta_1=(\beta_1^{(l)}|\,\beta_1^{(r)}) \in \piU$, then $\beta_{1}^{(l)} = (2r, \dots, r+1)$ and $\beta_{1}^{(r)} = \alpha^{(r)}$. If $\beta_2 = (\beta_{2}^{(l)} | \, \beta_{2}^{(r)})\ne \beta_1$, then $\beta_2$ belongs to $\piC$ and is such that $\beta_{2} = \alpha + \epsilon$, for some $\epsilon$ in $E_R$. 
  Then  $\beta$ such that $\beta^{(l)} = (2r, \dots, r+1)$ and $\beta^{(r)} = \beta_2^{(r)}$ satisfies $\beta \leq \gamma$ and $\beta\in\piU$ is a cover of both $\beta_1$ and $\beta_2$.
\item $\alpha = -\infty$. Then $\beta_1=\beta_2=(r, \dots, 1 | 1, \dots, r)$ by lemma~\ref{lemma:minmax}.
\end{enumerate}
\end{proof}
As a corollary, as $\fq$ and $\fq[\xv]$ are both Cohen-Macaulay rings, we get
\begin{theorem}
  Let $\fq$ be a field. Then $\fq[\Um, \Cm_I] / \mathcal S$ and $\fq[\xv, \Um, \Cm_I] / \mathcal S$ are both Cohen-Macaulay ring.
\end{theorem}
We now compute its dimension. We refer to \cite[Chap. 9]{E95} for the definition of (Krull) dimension and codimension (or height) of an ideal.
\begin{lemma}\label{lemma:krulldimension}
  The Krull dimension of $\fq[\Um, \Cm_I] / \mathcal S$ is $r(m+n-r) + 1$.
\end{lemma}
\begin{remark}
  The Krull dimension of the Minors ideal is
  $\dim(\fq[\Um]/\mathcal D)=r(m+n-r)$ from \cite[Theorem
  3.4.6]{BCRV22}. Indeed, the variety associated to $\mathcal D$ is
  the one associated to $\mathcal S$ projected onto the $\Um$
  variables. Moreover, for any solution $\Um^*$, there is a unique projective point in $\mathbb P(\fq^{\binom{n}{r}})$ that defines the vector space generated by the rows of $\Um^*$.
\end{remark}
\begin{proof}
  From \cite[Proposition 5.10 p. 55]{BV06}, the Krull dimension of $\fq[\Um, \Cm_I] / \mathcal S$ is given by the maximal length of a strictly decreasing chain of elements of $\piS$ 
  \begin{align*}
    h_1 > \dots  > h_p
  \end{align*}	
  where $h_i \in \piS$ for all $1 \leq i \leq p$. For  $\piS$, lemmas~\ref{lemma:minmax},~\ref{lemma:cover} tell us that  such a chain of covers must respect the following rules:
  \begin{enumerate}
  \item it begins with $h_p=(r, \dots, 1 | 1, \dots, r)$, which is the smallest one.
  \item it ends with $h_1=(m+r | n)$, which is bigger than every minor in $\piS$.
  \item for two consecutive $h_i>h_{i+1}$, we are in one of the following cases:
    \begin{enumerate}
    \item if $h_{i+1}= (m+r,h^{(l)}|h^{(r)}, n)$ then $h_i=(h^{(l)}|h^{(r)})$,
    \item $h_{i+1}+\epsilon=h_i$ with $\epsilon\in E_R$,
    \item $h_{i+1}+\epsilon=h_i$ with $\epsilon\in E_L$,
    \item if $h_{i+1}=(2r, \dots, r+1|h^{(r)})\in\piU$ then $h_i=(r,\dots,1|h^{(r)})\in\piC$.
    \end{enumerate}
  \end{enumerate}
    Only case (a) decreases the length, so that this operation must be
    performed exactly $r$ times to go from $h_p$ to $h_1$. Case (d)
    will be applied exactly once, to switch from monomials
    $h_{i+1}\in\piU$ to $h_i\in\piC$.  Moreover, between two minors of
    the same length, we can just decrease one coefficient by one. For the coefficient corresponding to the $d$-th row ($d\in\{1,\dots,r\}$), we need to go from $m+r$ to $r+d$, so we need $m-d$ steps. For the coefficient corresponding to the $d$-th column, we need to go from $n$ to $d$, so we need $n-d$ steps.
  Any such  chain has necessarily length
\begin{align*}
 r + \sum\limits_{d=0}^{r-1} (n-d) + \sum\limits_{d=0}^{r-1} (m-d) + 1 = r(m+n-r) + 1.
\end{align*}
	Therefore, $\dim(\fq[\Um, \Cm_I] / \mathcal S) = r(m+n-r) + 1$.
\end{proof}

It follows from this lemma and the dimension $\dim(\fq[\Um,\Cm_I])=mn+r(n-r)+1$ (from~\eqref{eq:dimplucker}) the following formula :
\begin{corollary}
  The height of $\mathcal S$ is        
  \[ \text{height }(\mathcal S) =  m(n-r). \]   
\end{corollary}
\begin{proof}
  According to \cite[p.226]{E95}, for a finitely generated domain $R$
  over a field, and $I\subset R$ an ideal, we have
  $\text{codim}(I) = \text{height}(I)=\dim(R)-\dim(I)$. As the Plücker
  algebra is factorial (hence a domain), see~\cite[Thm
  6.2.9]{BCRV22}, so is $R=\fq[\Um,\Cm_I]$. We get 
  $\text{height}(\mathcal S)=\dim(\fq[\Um,\Cm_I])-\dim(\fq[\Um, \Cm_I]
  / \mathcal S)=mn+r(n-r)+1-r(m+n-r)-1 = m(n-r)$.
\end{proof}
\subsection{Adding \texorpdfstring{$mn$}{mn} generic linear forms  to \texorpdfstring{$\mathcal S_{d_c}$}{Sdc}}
\label{sec:mngeneric}
We have computed in the previous section the Hilbert series of the module ${\fq[\Um]_{d_c}/\mathcal S_{d_c}}$. In this section, following the proof from \citet{FSS13} with some adjustment, we compute the Hilbert series of the module where we add $mn$ generic forms of degree $D$ in the variables $\Um$ and new variables $\xv$. For that, we put a weight $D$ on each variable $u_{i,j}$ and consider the weighted  Hilbert series $w\HS_{\fq[\Um]_{d_c}/\mathcal S_{d_c}}(t)=\HS_{\fq[\Um]_{d_c}/\mathcal S_{d_c}}(t^D)$.

Consider new variables $\mathfrak{b} = \{ \mathfrak{b}_t^\ell : t \in Monomials(\fq[X], D), 1 \leq l \leq mn \}$ and $\mathfrak{e} = \{ \mathfrak{e}_{i, j}^\ell : 1 \leq i \leq m, 1 \leq j \leq n, 1 \leq l \leq mn \}$. Denote by $\mathcal H_{d_c}=\fq(\mathfrak{b,e})[\Um,\xv]_{d_c}$, and define, for $\ell\in\{1,\dots,mn\}$:
\begin{align*}
  g_\ell = \sum\limits_{t \in Monomials(\fq[X], D)} \mathfrak{b}_t^\ell t + \sum\limits_{1 \leq i \leq m, 1 \leq j \leq n} \mathfrak{e}_{i, j}^\ell u_{i, j}
\end{align*}
and
\begin{align*}
\tilde{\mathcal{S}_{d_c}} = \mathcal{S}_{d_c} + \langle g_1, \dots, g_{mn}\rangle
\end{align*}
where $\langle g_1, \dots, g_{mn}\rangle = \langle g_1,\dots,g_{mn}\rangle_{\fq(\mathfrak{b,e})[\Um,\xv]}\mathcal H_{d_c}$ is a submodule of $\mathcal H_{d_c}$.

As Hilbert  series are invariant if we replace the field of coefficients $\fq$ by  an extension field $\fq(\mathfrak{b},\mathfrak{e})$, and are just divided by $(1-t)^K$ if we add $K$ variables $\xv=(x_1,\dots,x_K)$ to the polynomial ring (new independent variables are non-zero divisors), we have
\begin{align*}
  w\HS_{\mathcal H_{d_c}/\mathcal S_{d_c}}(t) &=  \HS_{\fq[\Um]_{d_c}/\mathcal S_{d_c}}(t^D)\frac{1}{(1-t)^K}
\end{align*}
To compute the Hilbert series for $\tilde{\mathcal S_{d_c}}$, we just have to show that $(g_1,\dots,g_\ell)$ is a regular sequence in $\mathcal H_{d_c}/\mathcal S_{d_c}$.

\begin{proposition}\label{prop:zerodivisor}
	Let $1 \leq \ell \leq mn$. If $g_\ell$ divides zero in $\fq(\mathfrak b, \mathfrak e)[\xv, \Um, \Cm_I] / \left( \mathcal S + \langle g_1, \dots, g_{\ell-1} \rangle \right)$, then there exists a prime ideal $P$ associated to $\mathcal S + \langle g_1, \dots, g_{\ell-1} \rangle$ such that dim$(P) \leq r(n-r) + 1$.
\end{proposition}
\begin{proof}
	If $g_{\ell}$ divides zero in $\fq(\mathfrak b, \mathfrak e)[\xv, \Um, \Cm_I] / \left( \mathcal S + \langle g_1, \dots, g_{\ell-1} \rangle \right)$, then there exists a prime ideal $P$ of $\fq(\mathfrak b, \mathfrak e)[\xv, \Um, \Cm_I]$ associated to $\mathcal S + \langle g_1, \dots, g_{\ell-1} \rangle$ such that $g_{\ell} \in P$ (see \citet[Prop. 1.2.1]{BH98}). For $\ell \leq mn$, let $\mathfrak b^{\leq \ell}$ and $\mathfrak e^{\leq \ell}$ denote the sets 
	\begin{align*}
		\mathfrak b^{\leq \ell} &= \{ \mathfrak b_t^s | t \in Monomials(\fq[\xv], D), 1 \leq s \leq \ell \}\\
		\mathfrak e^{\leq \ell} &= \{ \mathfrak e_{i, j}^s | 1 \leq i \leq m, 1 \leq j \leq n, 1 \leq s \leq \ell \}
	\end{align*}
	Since $\mathcal S + \langle g_1, \dots, g_{\ell-1} \rangle$ is an ideal of $\fq(\mathfrak b^{\leq \ell -1}, \mathfrak e^{\leq \ell -1})[\xv, \Um, \Cm_I]$ and $P$ is an associated prime, there exists a Gröbner basis  $G_P$ of $P + \plucker$ which is a finite subset of $\fq(\mathfrak b^{\leq \ell-1}, \mathfrak e^{\leq \ell-1})[\xv, \Um, c_I]$.
	
	Let $NF_P$ denote the normal form associated to this Gröbner basis. Since $g_{\ell} \in P$, we have $NF_P(g_{\ell}) = 0$ and by linearity of $NF_P$ we obtain :
        \begin{align}
          \sum\limits_{t \in Monomials(\fq[\xv], D)} \mathfrak b_t^\ell NF_P(t) + \sum\limits_{1 \leq i \leq m, 1 \leq j \leq n} \mathfrak e_{i, j}^\ell NF_P(u_{i, j}) = 0.\label{eq:bell}
        \end{align}
	Since $G_P \subset \fq(\mathfrak b^{\leq \ell -1 }, \mathfrak e^{\leq \ell -1})[\xv, \Um, c_I]$ for any monomial $t$, $NF_P(t) \in \fq(\mathfrak b^{\leq \ell -1}, \mathfrak e^{\leq \ell-1})[\xv, \Um, c_I]$. Therefore, by considering~\eqref{eq:bell} as a polynomial in the variables $\mathfrak b^\ell$ and $\mathfrak e^\ell$, all its coefficients must be zero, i.e.  $NF_P(t) = NF_P(u_{i, j}) = 0$ for any $t \in Monomials(\fq[\xv], D), 1 \leq i \leq m, 1 \leq j \leq n$. Then any $u_{i, j} \in \Um$ and $t \in Monomials(\fq[\xv], D)$ is in $P + \plucker$ as an ideal of $\fq(\mathfrak b, \mathfrak e)[\xv, \Um, c_I]$ and then any $u_{i, j} \in \Um$ and $t \in Monomials(\fq[\xv], D)$ is in $P$ as an ideal of $\fq(\mathfrak b, \mathfrak e)[\xv, \Um, \Cm_I]=\fq(\mathfrak b, \mathfrak e)[\xv, \Um, c_I]/\plucker$.
	
	Therefore each variables $u_{i, j}$ and $x_k$ is in $P$ because it is a prime ideal. It follows that
        \begin{align*}
          \text{dim}(P) \leq \text{dim} \left( \fq(\mathfrak b, \mathfrak e)[\xv, \Um, \Cm_I] / \langle \xv, \Um \rangle \right)
        \end{align*}
	where $\fq(\mathfrak b, \mathfrak e)[\xv, \Um, \Cm_I] / \langle \xv, \Um \rangle \simeq \fq(\mathfrak b, \mathfrak e)[\Cm_I]$ is the Plücker algebra, which is of dimension $r(n-r)+1$.
\end{proof}

\begin{lemma}
  For all $1 \leq \ell \leq mn$ and $K \geq m(n-r)$, $g_{\ell}$ does not divide zero in
  \begin{align*}
    \fq(\mathfrak b, \mathfrak e)[\xv, \Um, \Cm_I] / \left( \mathcal S + \langle g_1, \dots, g_{\ell-1} \rangle \right)
  \end{align*}
  and $\dim(\mathcal S + \langle g_1, \dots, g_{\ell} \rangle) = K + r(m+n-r) + 1 - \ell$.
\end{lemma}
\begin{proof}
  We prove it by induction on $\ell$. The case $\ell = 0$ follows
  directly from lemma~\ref{lemma:krulldimension}. Now let us suppose
  that $\mathcal S + \langle g_1, \dots, g_{\ell-1} \rangle$ has
  dimension $K + r(m+n-r)+2 - \ell$. Since
  $\fq(\mathfrak b, \mathfrak e)[\xv, \Um, \Cm_I] / \mathcal S$ is
  Cohen-Macaulay and $\langle g_1, \dots, g_{\ell-1} \rangle$ has
  codimension $\ell-1$ by the induction assumption, the Macaulay
  unmixed theorem from \citet[Corollary 18.14 p.456]{E95} says that each associated prime of
  $\mathcal S + \langle g_1, \dots, g_{\ell-1} \rangle$ has the same
  dimension. Now let us suppose that $g_{\ell}$ divides zero in
  $\fq(\mathfrak b, \mathfrak e)[\xv, \Um, \Cm_I] / \left( \mathcal S
    + \langle g_1, \dots, g_{\ell-1} \rangle \right)$. By proposition~\ref{prop:zerodivisor},
  there exists an associated prime $P$ of
  $\mathcal S + \langle g_1, \dots, g_{\ell-1} \rangle$ of dimension
  less than or equal to $r(n-r) + 1$, which is impossible because
  $\mathcal S + \langle g_1, \dots, g_{\ell-1} \rangle$ has dimension
  $K + r(m+n-r) - \ell + 2 \ge r(n-r)+2$ as $K \geq m(n-r)$, $\ell \leq mn$.
\end{proof}
The property is still valid for each $d_c\ge 1$ on the modules
$\fq(\mathfrak b, \mathfrak e)[\xv, \Um]_{d_c}$. Remember that
$\langle g_1,\dots,g_\ell\rangle$ is different depending on the structure (as an ideal in $\fq(\mathfrak b, \mathfrak e)[\xv, \Um,\Cm_I]$, as a module in $\fq(\mathfrak b, \mathfrak e)[\xv, \Um]_{d_c}$).
\begin{proposition}
  For all $1 \leq \ell \leq mn$ and $K \geq m(n-r)$, $g_{\ell}$ does not divide zero in
  \begin{align*}
    \fq(\mathfrak b, \mathfrak e)[\xv, \Um]_{d_c} / \left( \mathcal S_{d_c} + \langle g_1, \dots, g_{\ell-1} \rangle \right).
  \end{align*}
\end{proposition}

\begin{proof}
	We consider the following application :
	\begin{align*}
		\varphi : \fq(\mathfrak b, \mathfrak e)[\xv, \Um]_{d_c} \rightarrow \fq(\mathfrak b, \mathfrak e)[\xv, \Um, \Cm_I] / \left( \mathcal S + \langle g_1, \dots, g_{\ell-1} \rangle \right)		
	\end{align*}
	
	Firstly, $\left( \mathcal S + \langle g_1, \dots, g_{\ell-1} \rangle \right) \cap  \fq(\mathfrak b, \mathfrak e)[\xv, \Um]_{d_c} = \mathcal S_{d_c} + \langle g_1, \dots, g_{\ell-1} \rangle \cap \fq(\mathfrak b, \mathfrak e)[\xv, \Um]_{d_c}$ so we can factorize $\varphi$ and obtained an injection :
	
	\begin{align*}
		\fq(\mathfrak b, \mathfrak e)[\xv, \Um]_{d_c} / \left( \mathcal S_{d_c} + \langle g_1, \dots, g_{\ell-1} \rangle \right) \hookrightarrow \fq(\mathfrak b, \mathfrak e)[\xv, \Um, \Cm_I] / \left( \mathcal S + \langle g_1, \dots, g_{\ell-1} \rangle \right)
	\end{align*}
	It follows that if $g_{\ell}$ is not a zero-divisor in $\fq(\mathfrak b, \mathfrak e)[\xv, \Um,\Cm_I] / \left( \mathcal S + \langle g_1, \dots, g_{\ell-1} \rangle \right)$ then $g_\ell$ does not divide zero in $\fq(\mathfrak b, \mathfrak e)[\xv, \Um]_{d_c} / \left( \mathcal S_{d_c} + \langle g_1, \dots, g_{\ell-1} \rangle \right)$.
\end{proof}

\begin{corollary}\label{cor:Sdctilde}
  For $K\ge m(n-r)$ we have
  \begin{align*}
    w\HS_{\mathcal H_{d_c}/\tilde{\mathcal S_{d_c}}}(t) &=  \HS_{\fq[\Um]_{d_c}/\mathcal S_{d_c}}(t^D)\frac{(1-t^D)^{mn}}{(1-t)^K}.
  \end{align*}
\end{corollary}

 \subsection{Adding generic polynomials \texorpdfstring{$u_{i,j}-f_{i,j}$}{uij-fij}}
 \label{sec:uij}
 Let us consider the sets of variables
 \begin{align*}
   \mathfrak{a} = \{\mathfrak{a}_{i, j, t} : 1 \leq i \leq m, 1 \leq j \leq n, t \in Monomials(\fq[X], D) \}.
 \end{align*}
Let us define, for all $1 \leq i \leq m, 1 \leq j \leq n$ and $1 \leq l \leq mn$,
\begin{align*}
  f_{i, j} = \sum\limits_{t \in Monomials(\fq[X], D)} \mathfrak{a}_{i, j, t} t
\end{align*}
and
$\tilde{\idealSM_{d_c}} = \mathcal{S}_{d_c} + \langle u_{i, j} - f_{i, j}\rangle$ the submodule of $\fq(\mathfrak{a})[\Um,\xv]_{d_c}$, and $\idealSM_{d_c}=\tilde{\idealSM}_{d_c}\cap\fq(\mathfrak a)[\xv]$ the ideal SM with a matrix $\Fm$ whose entries are polynomials with generic coefficients $\mathfrak a$. 
Under  the hypothesis that $\fq$ is algebraically closed (we then denote it by $\overline{\mathbb{K}}$), we have the following genericity proposition.
We denote by $\mathbb A^{N_1}$ (resp. $\mathbb A^{N_2}$) an affine space of dimension $N_1=mn\binom{K+D-1}{D}$ (resp. $N_2=mn(\binom{K+D-1}{D}+mn)$) over $\fq$. We also reuse from~\cite{FSS13} the notation $\varphi_a$ (resp. $\psi_{b,e}$) for the evaluation morphisms that evaluate polynomials and ideals in $a$ on the variables $\mathfrak a$ (resp. in $b,e$ for the variables $\mathfrak{b,e}$).
\begin{proposition}
  There exist a non-empty Zariski open subset $\mathcal{O} \subset \mathbb A^{N_1}$ such that for all $a \in \mathcal{O}$ we have:
  \begin{align*}
    \HS_{\fq[\xv]_{d_c} / \varphi_a(\idealSM_{d_c})}(t)
    =\HS_{\fq(\mathfrak a)[\xv]_{d_c} / \idealSM_{d_c}}(t)
    = w\HS_{\mathcal H_{d_c} / \tilde{\mathcal{S}_{d_c}}}(t).
  \end{align*}
\end{proposition}
\begin{proof}
  The proof is very similar to that in~\cite{FSS13}.
  First, notice that the polynomials $u_{i,j}-f_{i,j}$ are clearly non-zero divisors in $\fq(\mathfrak a)[\Um,\xv]_{d_c}/(\mathcal S_{d_c} + \langle u_{i',j'}-f_{i',j'}\rangle_{(i',j')\ne (i,j)})$. As the number of these polynomials is the same as the number of variables $\Um$, we have
  \begin{align*}
    \HS_{\fq(\mathfrak a)[\xv]_{d_c} / \mathcal{{I}}_{d_c}}(t)
    &=w\HS_{\fq(\mathfrak a)[\Um,\xv]_{d_c} / \mathcal{\tilde{I}}_{d_c}}(t).
  \end{align*}
  The same is true for any evaluation on $a\in\mathbb A^{N_1}$:
  \begin{align}\label{eq:1}
    \HS_{\fq[\xv]_{d_c} / \varphi_a(\mathcal{{I}}_{d_c})}(t)
    =w\HS_{\fq[\Um,\xv]_{d_c} / \varphi_a(\mathcal{\tilde{I}}_{d_c})}(t),
  \end{align}
 We fix an admissible monomial ordering on $\fq(a)[\Um,\xv]_{d_c}$. Following the proof of~\cite[Lemma 4]{FSS13}, using the properties of Gröbner basis on modules (see e.g.~\cite{BV06}),  there exists a non-empty Zariski open set $\mathcal O_1\subset\mathbb A^{N_1}$ such that, for all $a\in O_1$,
 the modules generated by the leading monomials in $\tilde{\idealSM_{d_c}}$ and $\varphi_a(\tilde{\idealSM_{d_c}})$ are the same, hence their (weighted) Hilbert series are the same:
 \begin{align*}
  w\HS_{\fq[\Um,\xv]_{d_c} / \varphi_a(\mathcal{\tilde{I}}_{d_c})}(t)
=   w\HS_{\fq(\mathfrak{a})[\Um,\xv]_{d_c} / \mathcal{\tilde{I}}_{d_c}}(t),& \forall a\in O_1.
 \end{align*}
 This proves the first equality. The same kind of proof shows that there exists a non-empty Zariski open set $O_2\subset \mathbb A^{N_2}$ such that
 \begin{align}\label{eq:2}
  w\HS_{\fq[\Um,\xv]_{d_c} / \psi_{(b,e)}(\mathcal{\tilde{S}}_{d_c})}(t)
=   w\HS_{\fq(\mathfrak{b,e})[\Um,\xv]_{d_c} / \mathcal{\tilde{S}}_{d_c}}(t),& \forall (b,e)\in O_2.
 \end{align}
 We then use~\cite[Lemma 3]{FSS13}, that is still valid in our context, and says that there exists a 
 non-empty Zariski open set $O_3\subset \mathbb A^{N_1}$, such that for any $a\in O_3$, there exists $(b,e)\in O_2$ such that $\varphi_a(\mathcal{\tilde{I}}_{d_c}) = \psi_{(b,e)}(\mathcal{\tilde{S}}_{d_c})$. Now consider some $a\in O_1\cap O_3$ that is non-empty, and the  $(b,e)\in O_2$ such that $\varphi_a(\mathcal{\tilde{I}}_{d_c}) = \psi_{(b,e)}(\mathcal{\tilde{S}}_{d_c})$, then we get:
 \begin{align*}
   \HS_{\fq[\xv]_{d_c} / \varphi_a(\mathcal{{I}}_{d_c})}(t)
   & =w\HS_{\fq[\Um,\xv]_{d_c} / \varphi_a(\mathcal{\tilde{I}}_{d_c})}(t)&(\text{by } \eqref{eq:1}),\\
   &  =w\HS_{\fq[\Um,\xv]_{d_c} / \psi_{(b,e)}(\mathcal{\tilde{S}}_{d_c})}(t)\\
&=   w\HS_{\fq(\mathfrak{b,e})[\Um,\xv]_{d_c} / \mathcal{\tilde{S}}_{d_c}}(t)&(\text{by } \eqref{eq:2}).
 \end{align*}
 This concludes the proof by definition of $\mathcal H_{d_c}$.
 
\end{proof}
The Hilbert series for $\mathcal H_{d_c}/\tilde{\mathcal{S}}_{d_c}$ is known for $K\ge m(n-r)$ according to corollary~\ref{cor:Sdctilde}, and our experiments presented in Section~\ref{sec:experiments} suggest that  it remains true for any $K\le (m-r)(n-r)$, and for $K\in\{(m-r)(n-r)+1,\dots,m(n-r)-1\}$ provided that $d_c\le m-r$. When this hypothesis is satisfied, the Hilbert series for the $\SM$ system in degree $d_c \geq 1$ in $\Cm$ is given by the formula~\eqref{eq:HSAfinal}.

\section{Application to the Mirath signature scheme}
\label{sec:complexity}
In this section we analyze the complexity of solving the Support Minors system for the Mirath signature scheme, whose security rely on the hardness of uniformly random overdetermined MinRank instances with a unique solution. The parameters chosen for Mirath over $\mathbb F_{16}$ are given in table~\ref{tab:mirath} and are such that  $K<(m-r)(n-r)$.

Over finite fields, the total number of possible systems is finite, and \emph{generic} means that the proportion of systems satisfying a given property is large.
The MinRank instances for Mirath are uniformly generated by design, hence we expect these instances to be generic and use the Hilbert series from this paper to analyze these systems. 
The genericity over a finite field has been discussed and experimentally verified in~\cite{BBCGPSTV20}, and confirmed by our own experiments.
  In particular, for any $d_c$ we can compute a degree\footnote{Note that $d_x^{reg}$ depends on $d_c$.} of regularity $d_x^{reg}$ as  the degree of the Hilbert series~\eqref{eq:HSB} plus one.
  The only difference with really random instances is that the Macaulay matrix of degree $(d_x^{reg},d_c)$ has co-dimension 1 instead of 0. As explained in~\cite[Section 5.8]{BBCGPSTV20},  it is enough to find a non-zero vector in the right kernel of this matrix to get the solution of the system.

  In our context, the number of columns of the Macaulay matrix in bi-degree $(d_x,d_c)$ is given by the number of monomials in this degree, which is
  \begin{align*}
{\tt M}(d_x,d_c) =    \binom{K+d_x-1}{d_x}\det\left(\binom{n+d_c-i}{n-j}_{1\le i,j\le r}\right)
  \end{align*}
  following equation~\eqref{eq:rkCm}.   Its rank is ${\tt M}(d_x,d_c)-1$.
  Its rows correspond to the polynomials~\eqref{eq:SM} multiplied by all the monomials of degree $(d_x-D,d_c-1)$, and its number may be larger than ${\tt M}(d_x,d_c)$.

However, the $F_5$ algorithm from~\citet{F02} contains a criterion that allows to construct submatrices of the Macaulay matrix that have full rank for regular sequences.
Even if such a criterion has not been designed yet for the Support Minors modeling, 
we can expect such results to arise in the coming years (for instance, the syzygies of the system are already known for $d_c=1$ and $d_x\le r+1$, see~\cite{BBCGPSTV20}). See the work~\cite{GNS24} for the Minors system for instance, or~\cite{CGG25} that provides proofs about the syzygies of the SM system for $r=n-1$ or $b=2$. We assume in the following that it is possible to construct a sub-matrix of the Macaulay matrix in degree $(d_x^{reg},d_c)$ that has the same number of rows as its rank.

To compute the solution, we use linear algebra. We can compute an echelon form of the matrix, or use the Wiedemann algorithm~\cite{W86} to compute directly an element in the kernel, in which case the complexity is expressed in terms of the density ${\tt D}_{d_c}$ of the matrix.
Each polynomial in $(\SM)$ has at most ${\tt D}_1=K(r+1)$ non-zero coefficients for $d_c=1$. For $d_c>1$, we first need to multiply the polynomials by monomials in the $c_I$ variables, then apply the straightening law to express the polynomials in the basis of the standard monomials. We do not take into account the cost of this reduction,  it will be smaller than the cost of linear algebra on the Macaulay matrix.
We do not have a bound on the final number of terms, and consider that all the monomials can be present in each equation, i.e. ${\tt D}_{d_c}\le {\tt M}(1,d_c)$ for $d_c>1$.
  The complexity is given by
  \begin{align}
\mathbb C = \min\left(c_\omega {\tt M}(d_x^{reg},d_c)^\omega,
    c{\tt D}_{d_c}{\tt M}(d_x^{reg},d_c)^2\right).\label{eq:wiedemann}
  \end{align}
  where $\omega$ is the exponent of matrix multiplication and $c_\omega$ the constant involved. For instance we can take $\omega=2.81$ and $c_\omega=3$~\cite{JPS13}.
  For the Wiedemann algorithm, $c$ is a small constant.

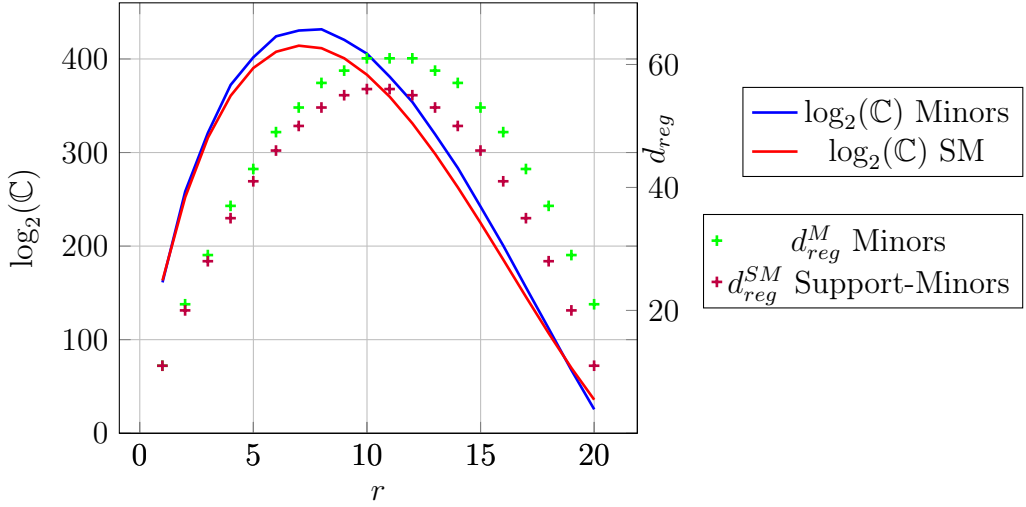
\begin{figure}[t]
  \centering
  \begin{tikzpicture}
    \begin{axis}[grid= major,
      xlabel = {$r$},
      ylabel = {$\log_2(\mathbb C)$},
      ymin=0, ymax=460, 
      legend entries={$\log_2(\mathbb C)$ Minors, $\log_2(\mathbb C)$ SM},
            legend style={at={(1.745,0.7)},anchor=east}]
      \addplot[mark=none, draw=blue, line width=1pt] coordinates {(1,161.2) (2, 258.4) (3, 321.1) (4, 372.3) (5, 401.8) (6, 424.2) (7, 430.4) (8, 431.7) (9, 420.4) (10, 405.6) (11, 381.0) (12, 353.9) (13, 319.4) (14, 283.4) (15, 242.1) (16, 200.4) (17, 155.7) (18, 111.7) (19, 67.3) (20, 25.5)}; 
      \addplot[mark=none, draw=red, line width=1pt] coordinates {(1, 163.3) (2, 252.3) (3, 315.9) (4, 360.7) (5, 390.4) (6, 407.7) (7, 414.2) (8, 411.5) (9, 400.7) (10, 383.1) (11, 359.6) (12, 331.1) (13, 298.5) (14, 262.7) (15, 224.7) (16, 185.4) (17, 145.6) (18, 106.7) (19, 69.6) (20, 35.8)}; 
    \end{axis}
    \begin{axis}[
      axis y line*=right,
      axis x line*=none,
      ymin=0, ymax=70, 
      ylabel = {$d_{reg}$},
      ylabel style={at={(1.04,0.6)}, anchor=west}, 
      ytick={20,40,60}, 
      legend entries = {$d^M_{reg}$ Minors, $d_{reg}^{SM}$ Support-Minors},
      legend style={at={(1.745,0.4)},anchor=east}]
      \addplot[mark=+, draw= green, only marks, line width=1pt] coordinates {(1, 11) (2, 21) (3, 29) (4, 37) (5, 43) (6, 49) (7, 53) (8, 57) (9, 59) (10, 61) (11, 61) (12, 61) (13, 59) (14, 57) (15, 53) (16, 49) (17, 43) (18, 37) (19, 29) (20, 21)};
      \addplot[mark=+, draw= purple, only marks, line width=1pt] coordinates {(1, 11) (2, 20) (3, 28) (4, 35) (5, 41) (6, 46) (7, 50) (8, 53) (9, 55) (10, 56) (11, 56) (12, 55) (13, 53) (14, 50) (15, 46) (16, 41) (17, 35) (18, 28) (19, 20) (20, 11)};
    \end{axis}
  \end{tikzpicture}
  \caption{For $n = m = 22$, we plot for each value $r\in\{1,\dots,20\}$ and $K=(m-r)(n-r)-1$ the degree of regularity of the generic system for the Minors and SM modelings at $d_c=1$ systems (crosses), as well as the
log of the    complexity $\mathbb C$ given by equation~\eqref{eq:wiedemann} (lines).}
  \label{fig:paraMIRAth}
\end{figure}

We plot in Figure~\ref{fig:paraMIRAth} the best complexities for different  values of $r$, to see the impact of $r$ on the complexity. We chose $m=n=22$, which are the parameters for the Mirath signature scheme over  $\mathbb F_{16}$, and take $K=(m-r)(n-r)-1$. We plot the complexity evaluated with the formula~\eqref{eq:wiedemann} for the Minors and the $\SM$ systems at $d_c=1$, as well as the degree of regularity. We can see that the $\SM$ system behaves better than the Minors one for almost all values of $r$, and that the selected value $r=6$ for Mirath is close to the hardest one.

Over a finite field $\mathbb F_q$, it is always possible to perform a hybrid approach~\cite{BMT23} to reduce the cost of solving $(\SM)$ for parameters $(m,n,K,r)$ to the cost of solving $q^{ar}$ systems of parameters $(m,n-a,K-am,r)$. The security for the Mirath signature scheme was estimated previously  using the MinRank
 CryptographicEstimators V2.0.0\footnote{Code available at \url{https://github.com/Crypto-TII/CryptographicEstimators}.} only for $d_c=1$ and $d_{reg}\le r+1$. With our results, we can extend those estimates to any $d_x$ and $d_c$. The results are given in table~\ref{tab:mirath}: even if we can actually find parameters such that the resulting complexity is better, the gain is only a few bits and for $d_c=1$. This supports the belief that the parameters for Mirath are well chosen. Note that we add 4 bits of complexity to take into account the bit complexity $O(\log_2(q)^2)$ of a field operation in $\mathbb F_{16}$.
 
\begin{table}[t]
  \centering
  \begin{tabular}{|c|c|c|c|c|c|c|c|}
    \hline
    Security & parameters & \multicolumn{3}{c|}{previous results}& \multicolumn{3}{c|}{this paper}\\
    level & $(m,n,K,r)$ & $\log_2$(compl.)   & $d_{reg}$ & $a$& $\log_2$(compl.) & $d_{reg}$ & $a$\\
    \hline
    I & (16, 16, 143, 4) & 166  & 2 & 8 & 164 & 6& 5\\
    \hline
    III & (19, 19, 195, 5) & 227 & 6 & 7 & 227 & 6 & 7\\
    \hline
    V & (22, 22, 255, 6) & 301  & 1 & 11& 298 & 10 & 7\\
    \hline
  \end{tabular}
  \caption{Parameters and security for the Mirath signature scheme. The security level I (resp. III, V) corresponds to a required security of $2^{142}$ (resp. $2^{206}$, $2^{270}$) bit operations.  The best complexity is always at $d_c=1$.}
  \label{tab:mirath}
\end{table}

\section{Experimental results}
\label{sec:experiments}
We have verified experimentally that the
  formula \eqref{eq:HSAfinal} correctly predicts the Hilbert
  series, for small values of the parameters. We limited ourselves to
  $D=1$ for practical reasons. We used the {\tt magma} computer algebra system~\cite{BCP97}.

  We use a naive algorithm: we compute the initial
  equations~\eqref{eq:SM}, multiply them by all monomials to get
  equations of bi-degree $(d_x,d_c)$, reduce the equations by the
  Plücker relations and compute the dimension of the $\fq$-vector
  space generated by the equations in bi-degree $(d_x,d_c)$. We deduce
  from that the coefficient of $t^{d_x}$ of the Hilbert
  series~\eqref{eq:HSAfinal}.

We ran experiments for $q=31$, $(m,n,r)=(5,5,2)$, $(5,5,3)$, $(6,5,2)$,
$(6,5,3)$, $(7,6,2)$, $(7,6,3)$, $K\in\{2,\dots,mn\}$, $d_c\in\{1,\dots,6\}$. 
The Hilbert series was the predicted one for all the proven cases ($K\ge m(n-r)$), but also in the overdetermined cases $K\le (m-r)(n-r)$ (up to a $[]_+$ truncation of the series) and for all cases when $d_c\le m-r$. 

Note that for $K\le (m-r)(n-r)$, for $d_c=m-r$ the Hilbert series is a
constant,
so that if the formula is correct for $d_c=m-r$ it remains true for
$d_c>m-r$.

However, for $K\in\{(m-r)(n-r)+1,\dots,(m-r)r-1\}$ and $d_c> m-r$, the
Hilbert series was different from the predicted one. We do not have
yet an explanation for this phenomenon, but this appears when the
determinant $\det(A_{d_c})$ has negative coefficients and the Hilbert
series has a non-trivial denominator. For instance, for $m=n=K=5$ and
$r=3$, for $d_c=3$ the predicted Hilbert series is
$\HS(t)=\frac{175-175t+50t^2}{(1-t)}$ but the correct one is
$\HS(t)=\frac{175-125t}{(1-t)}$. For $d_c\le m-r$, the determinant
$\det(A_{d_c})$ has only positive coefficients.

The probability of a system to be generic over a finite field was
already experimented in~\cite{BBCGPSTV20}. We ran our own tests over
1000 random instances with the same parameters but larger $d_c$, and
got the same kind of results: most of the time, the rank of the
matrices was as expected. However, in cases where the matrices are
almost square, we had non-generic systems. For instance for
$m=6, n=5$, $K=6$, $r=3$, $d_c\in\{1,2,3\}$, we had $58\%$ of generic
systems over $\mathbb F_4$ and $96\%$ over $\mathbb F_{31}$, for very
small matrices (the same systems are generic for all $d_c$ or are
not).
\section*{Acknowledgments}
This work received funding from the France 2030 program, managed by
the French National Research Agency under grant agreement
No. ANR-22-PETQ-0008 PQ-TLS.

\begin{table}[t]
  \centering
  \begin{tabular}{|c|p{10.9cm}|}
    \hline
    $\Fm = (f_{i,j})$ & a $m\times n$ matrix whose entries are homogeneous polynomials $f_{i,j}\in\fq[\xv]$ of degree $D$\\
    \hline
    $\Cm = (c_{i,j})$ & a $r\times n$ matrix whose entries are variables in $\fq[\Cm]$ \\
    \hline
    $\Um = (u_{i,j})$ & a $m\times n$ matrix whose entries are variables, the entries belong to $\fq[\Um]=\fq[(u_{i,j})_{i,j}]$\\
    \hline
    \hline
    $\fq[\Ym]_{d_c}$ & Sub $\fq[\Ym]$-module of $\fq[\Ym, \Cm_I]$ generated by the monomials in $\Cm_I$ of degree $d_c$ \\
    \hline
    \hline
    $\idealSM$ & Ideal of $\fq[\xv, \Cm_I]$ generated by the Support-Minors equations~\eqref{eq:SM}\\
    \hline
    $\idealSM_{d_c} = \idealSM \cap \fq[\xv]_{d_c}$ & Submodule of $\idealSM\subset\fq[\xv,\Cm_I]$ generated by the Support-Minors equations~\eqref{eq:SM} in degree $d_c$ in the $\Cm_I$\\
    \hline
    $\mathcal J$ & Ideal of $\fq[\xv]$ generated by the $(r+1) \times (r+1)$ minors of $\Fm$\\
    \hline
    \hline
    $\mathcal S_{eq}$ & Set of the \emph{determinantal} Support-Minors equations, as elements of  $\fq[\Um,\Cm_I]\subset\fq[\Um,\Cm]$\\
    \hline
    $\mathcal I_{eq}$ & Ideal of $\fq[\Um, \Cm]$ generated by $\mathcal S_{eq}$\\
    \hline
    $\mathcal Y_{eq}$ & Set of standard monomials $Y=\gamma_1\dots\gamma_t$ with $\gamma_1 \in \Seq$, it is a basis of $\mathcal I_{eq}$ as a $\fq$-vector space.\\
    \hline
    $\mathcal S_{d_c},\; d_c\ge 1$ & Submodule of $\fq[\Um]_{d_c}$, $\mathcal S_{d_c} = \fq[\Um]_{d_c} \cap \mathcal I_{eq}$\\
    \hline
    $\mathcal D$ & Determinantal ideal of $\fq[\Um]$ generated by the $(r+1) \times (r+1)$ minors of $\Um$\\
    \hline
    $\mathcal D_{\Cm,\Um}$ & Determinantal ideal of $\fq[\Um, \Cm]$ generated by the $(r+1) \times (r+1)$ minors of $\binom{\Cm}{\Um}$\\
    \hline
    $\mathcal S$ & Ideal of $\fq[\Um, \Cm_I]$ generated by $\mathcal S_{eq}$ and $\mathcal D$. According to lemma~\ref{lemma:DHinS}, $\mathcal S\cap \fq[\Um]_{d_c}=\mathcal S_{d_c}$.\\
    \hline
    \hline
    $\piS$ & Set of minors which define $\fq[\Um, \Cm_I] / \mathcal S$ as an ASL\\
    \hline
    $\mathcal H_{d_c}$ & $\mathcal H_{d_c} = \fq(\mathfrak b, \mathfrak e)[\Um, \xv]_{d_c}$\\
    \hline
    $g_{\ell}$ & Generic polynomial in $\fq(\mathfrak b, \mathfrak e)[\Um, \xv]$\\
    \hline
    $f_{i, j}$ & Generic polynomial in $\fq(\mathfrak a)[\xv]$\\
    \hline
    $\tilde{\mathcal S}_{d_c}$ & Submodule of $\mathcal H_{d_c}$, $\tilde{\mathcal S}_{d_c} = \mathcal S_{d_c} + \langle g_1, \dots, g_{mn} \rangle$\\
    \hline
    $\tilde{\idealSM}_{d_c}$ & Submodule of $\fq(\mathfrak a)[\Um, \xv]_{d_c}$, $\tilde{\idealSM}_{d_c} = \mathcal S_{d_c} + \langle u_{i,j} - f_{i, j} \rangle$.\\
    \hline
  \end{tabular}
  \caption{Glossary for all structures, systems and ideals used in the paper.}
  \label{table:notations}
\end{table}

\appendix
\section{Appendix}
\label{sec:appendix}
We give an elementary proof of the following lemma, given by \citet{CH94} without proof:
\begin{lemma}
    \begin{align}
    \sum_{k\ge 1} (-1)^{i-k}\binom{i-1}{k-1}\binom{a-k}{b-k} = (-1)^{i-1}\binom{a-i}{b-1} \text{ for all } i\ge 1, a \in \mathbb C, b \in \mathbb N^*.\tag{\ref{eq:combi}}
  \end{align}
\end{lemma}
\begin{proof} We use formal power series in $x$.
  With the change of variable $a\leftarrow a-1, b\leftarrow b-1, i\leftarrow i-1, k\leftarrow k-1$, we have (with the convention $\binom{a}{k}=0$ for $k<0$, $k\in\mathbb Z$):
  \begin{align*}
    \eqref{eq:combi}&\iff   \sum_{k\ge 0} (-1)^{k}\binom{i}{k}\binom{a-k}{b-k} = \binom{a-i}{b} \text{ for all } i\ge 0, a \in \mathbb C, b \in \mathbb N.\\
    &\iff \sum_{b\ge 0} \sum_{k\ge 0} (-1)^{k}\binom{i}{k}\binom{a-k}{b-k}x^b = \sum_{b\ge 0}\binom{a-i}{b}x^b \text{ for all } i\ge 0, a \in \mathbb C\\
                    &\iff \sum_{k\ge 0}(-1)^{k}\binom{i}{k} \sum_{b\ge k} \binom{a-k}{b-k}x^b = (1+x)^{a-i} \text{ for all } i\ge 0, a \in \mathbb C
  \end{align*}
  But the left hand side is
  \begin{align*}
    \sum_{k\ge 0}(-1)^{k}\binom{i}{k} \sum_{b\ge k} \binom{a-k}{b-k}x^b
    &= \sum_{k\ge 0}(-1)^{k}\binom{i}{k} x^k \sum_{b\ge 0} \binom{a-k}{b}x^b 
      = \sum_{k\ge 0}(-1)^{k}\binom{i}{k} x^k(1+x)^{a-k}\\
    &= (1+x)^a \sum_{k\ge 0}(-1)^{k}\binom{i}{k} \left(\frac{x}{1+x}\right)^k = (1+x)^a \left(1-\frac x{1+x}\right)^i\\
    &= (1+x)^a\frac1{(1+x)^i}=(1+x)^{a-i}.
  \end{align*}
\end{proof}

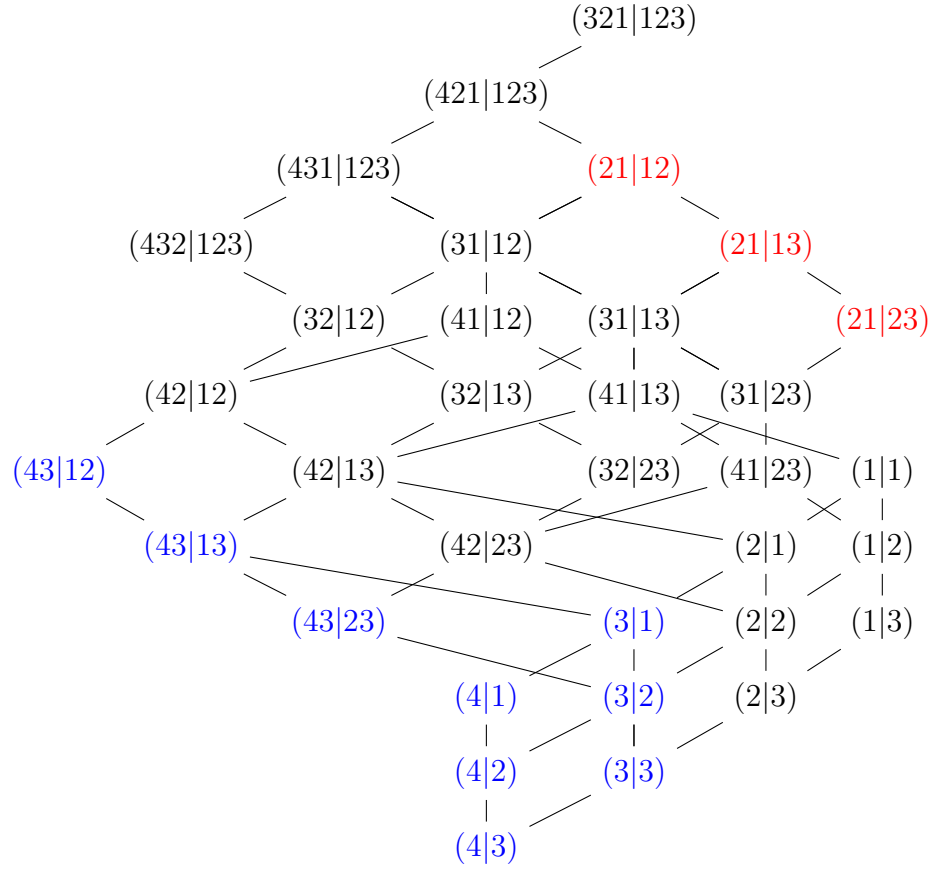
\begin{figure}[t]
  \centering
  \begin{tikzpicture}
  \matrix (mat) [row sep=0.3cm]
  { &&&&\node (a) {$(321|123)$};&&\\
    &&&\node (b) {$(421|123)$};&&\\
    &&\node (c1){$(431|123)$}; && \node(c2){\textcolor{red}{$(21|12)$}};\\
    &\node (d1){$(432|123)$}; && \node(d2){$(31|12)$};&& \node(d3){\textcolor{red}{$(21|13)$}};\\
    &&\node (e1){$(32|12)$}; & \node(e2){$(41|12)$}; & \node(e3){$(31|13)$};&&\node (e4){\textcolor{red}{$(21|23)$}};\\
    &\node(f1){$(42|12)$}; && \node(f2){$(32|13)$}; &\node(f3){$(41|13)$};&\node(f4){$(31|23)$};\\
    \node(g1){\textcolor{blue}{$(43|12)$}};&&\node(g2){$(42|13)$};&&\node(g3){$(32|23)$};&\node(g4){$(41|23)$};&\node(A1){$(1|1)$};\\
    &\node(h1){\textcolor{blue}{$(43|13)$}}; &&  \node(h2){$(42|23)$};&&\node(B1){$(2|1)$};&\node(B2){$(1|2)$};\\
&&    \node(i1){\textcolor{blue}{$(43|23)$}};&&\node(C1){\textcolor{blue}{$(3|1)$}};&\node(C2){$(2|2)$};&\node(C3){$(1|3)$};\\
&&&\node(D1){\textcolor{blue}{$(4|1)$}};&\node(D2){\textcolor{blue}{$(3|2)$}};&\node(D3){$(2|3)$};\\
&&&\node(E1){\textcolor{blue}{$(4|2)$}};&\node(E2){\textcolor{blue}{$(3|3)$}};\\
&&&\node(F1){\textcolor{blue}{$(4|3)$}};\\
  };
  \draw (a)--(b);
  \draw (b)--(c1)--(d2)--(c2)--(b);
  \draw (c1)--(d2)--(e1)--(d1)--(c1);
  \draw (c2)--(d2)--(e3)--(d3)--(c2);
  \draw (d2)--(e2)--(f3)--(e3)--(d2);
  \draw (d3)--(e3)--(f4)--(e4)--(d3);
  \draw (e1)--(f1)--(g1)--(h1)--(g2)--(f2)--(e1);
  \draw (e3)--(f3);
  \draw (e3)--(f2)--(g3)--(f4)--(e3);
  \draw (e2)--(f1)--(g2)--(f3)--(g4);
  \draw (g2)--(h2)--(i1)--(h1);
  \draw (g3)--(h2)--(g4)--(f4);
  \draw (f3)--(A1)--(B1)--(C1)--(D1)--(E1)--(F1);
  \draw (g4)--(B2)--(C2)--(D2)--(E2)--(F1);
  \draw (g2)--(B1)--(C2)--(D3);
  \draw (h2)--(C2);
  \draw (h1)--(C1)--(D2)--(E2);
  \draw (i1)--(D2)--(E1);
  \draw (A1)--(B2)--(C3)--(D3)--(E2);
  \end{tikzpicture}
  \caption{Hasse diagram for the minors for $m=2$, $n=3$, $r=2$. Elements in red are in \textcolor{red}{$\piC$}, in blue in \textcolor{blue}{$\piU$}. A line means that the element below is a cover for the one above.} 
  \label{fig:hasse}
\end{figure}

\end{document}